\documentclass[a4paper,UKenglish]{lipics-v2016}

\usepackage{microtype,paralist}
\usepackage[disable,textsize=footnotesize]{todonotes}

\title{Geodesic Obstacle Representation of Graphs\footnote{The work of PB, VD and PM is supported in part by NSERC. SM is supported by a Carleton-Fields postdoctoral fellowship.}}
\titlerunning{Geodesic Obstacle Representation of Graphs} 

\EventEditors{Editors}
\EventNoEds{2}
\EventLongTitle{Long Title}
\EventShortTitle{arXiv}
\EventAcronym{Acronym}
\EventYear{Event Year}
\EventDate{Event Dates}
\EventLocation{Event Location}
\EventLogo{}
\SeriesVolume{Volume}
\ArticleNo{No}

\usepackage{graphicx}
\usepackage{pifont}
\usepackage{verbatim}
\usepackage{amsmath,amsfonts}
\usepackage{enumerate,url}
\usepackage{paralist}
\usepackage{wrapfig}
\usepackage{multirow}
\usepackage{todonotes}
\usepackage{array,wrapfig}

\usepackage{algorithm}
\usepackage{algorithmicx}
\usepackage[noend]{algpseudocode}

\DeclareMathOperator{\ob}{obs}
\DeclareMathOperator{\planeobs}{plane-obs}

\newcommand{\eps}{\epsilon}
\newcommand{\R}{\mathbb{R}}

\makeatletter
    \renewcommand*{\@fnsymbol}[1]{\ensuremath{\ifcase#1\or *\or \dagger\or \S\or
       \mathsection\or \mathparagraph\or \|\or **\or \dagger\dagger
       \or \ddagger\ddagger \else\@ctrerr\fi}}
\makeatother

\author[1]{Prosenjit Bose}
	\affil[1]{School of Computer Science, Carleton University, Ottawa, Canada.\\
	\texttt{jit@scs.carleton.ca}}
\author[2]{Paz Carmi}
	\affil[2]{Department of Computer Science, Ben-Gurion University of the Negev, Beer-Sheva, Israel. \texttt{carmip@cs.bgu.ac.il}}
\author[3]{Vida Dujmovic}
	\affil[3]{School of Computer Science and Electrical Engineering, University of Ottawa, Ottawa, Canada. \texttt{vida.dujmovic@uottawa.ca}}
\author[4]{Saeed Mehrabi}
	\affil[4]{School of Computer Science, Carleton University, Ottawa, Canada.\\
	\texttt{saeed.mehrabi@carleton.ca}}
\author[5]{Fabrizio Montecchiani}
	\affil[5]{Department of Engineering, University of Perugia, Perugia, Italy.\\
	\texttt{fabrizio.montecchiani@unipg.it}}
\author[6]{Pat Morin}
	\affil[6]{School of Computer Science, Carleton University, Ottawa, Canada.\\
	\texttt{morin@scs.carleton.ca}}
\author[7]{Luis Fernando Schultz Xavier da Silveira}
	\affil[7]{School of Computer Science, Carleton University, Ottawa, Canada.\\
	\texttt{schultz@ime.usp.br}}

\authorrunning{Bose et al.} 

\Copyright{P. Bose, P. Carmi, V. Dujmovic, S. Mehrabi, F. Montecchiani, P. Morin, and L. Schultz}

\newcommand{\V}{\mathcal{V}}

\begin{document}

\maketitle

\begin{abstract}
An \emph{obstacle representation} of a graph is a mapping of the vertices onto points in the plane and a set of connected regions of the plane (called \emph{obstacles}) such that the straight-line segment connecting the points corresponding to two vertices does not intersect any obstacles if and only if the vertices are adjacent in the graph. The obstacle representation and its \emph{plane} variant (in which the resulting representation is a plane straight-line embedding of the graph) have been extensively studied with the main objective of minimizing the number of obstacles. Recently, Biedl and Mehrabi~\cite{BiedlM-GD17} studied \emph{grid obstacle representations} of graphs in which the vertices of the graph are mapped onto the points in the plane while the straight-line segments representing the adjacency between the vertices is replaced by the $L_1$ (Manhattan) shortest paths in the plane that avoid obstacles.

In this paper, we introduce the notion of \emph{geodesic obstacle representations} of graphs with the main goal of providing a generalized model, which comes naturally when viewing line segments as shortest paths in the Euclidean plane. To this end, we extend the definition of obstacle representation by allowing \emph{some} obstacles-avoiding shortest path between the corresponding points in the underlying metric space whenever the vertices are adjacent in the graph. We consider both \emph{general} and \emph{plane} variants of geodesic obstacle representations (in a similar sense to obstacle representations) under any polyhedral distance function in $\R^d$ as well as shortest path distances in graphs. Our results generalize and unify the notions of obstacle representations, plane obstacle representations and grid obstacle representations, leading to a number of questions on such embeddings.
\end{abstract}

\section{Introduction}
\label{sec:introduction}
An obstacle representation of an (undirected simple) graph $G$ is pair $(\varphi, S)$ where $\varphi:V(G)\to\R^2$ maps vertices of $G$ to distinct points in $\R^2$ and $S$ is a set of connected subsets of $\R^2$ with the property that, for every $u,w\in V(G)$, $uw\in E(G)$ if and only if the line segment with endpoints $\varphi(u)$ and $\varphi(w)$ is
disjoint from $\cup S$.  The elements of $S$ are called \emph{obstacles}. It is easy to see that every graph $G$ has an obstacle representation: obtain a straight-line drawing of $G$ by taking any $\varphi$ that does not map three vertices of $G$ onto a single line, and let $S$ be the set of the open faces in the resulting arrangement of line segments.

Since every graph has an obstacle representation, this defines a natural graph parameter called the \emph{obstacle number}, $\ob(G) = \min\{|S| :\text{$(\varphi, S)$ is an obstacle representation of $G$}\}$. Since their introduction by Alpert et al.~\cite{AKL10}, obstacle numbers have been studied extensively with the main goal of bounding the obstacle numbers of various classes of graphs (see e.g.~\cite{BalkoCV18,DujmovicM15,GimbelMV17,JohnsonS14,MukkamalaPP12,PachS11} and the references therein).

For planar graphs, there is also a natural notion of a \emph{plane obstacle representation} $(\varphi, S)$ which is an obstacle representation in which $\varphi$ defines a plane straight-line embedding of $G$. This leads to a \emph{plane obstacle number}: $\planeobs(G) = \min\{|S| :\text{$(\varphi, S)$ is a plane obstacle}$ $\text{representation of $G$}\}$. Using Euler's formula, it is not hard to see that the plane obstacle number of any $n$-vertex planar graph is $O(n)$: let $\varphi$ define any plane drawing of $G$ with no three vertices collinear and take $S$ to be the set of open faces in this drawing. Since an $n$-vertex planar graph has at most $2n-4$ faces, this implies $\planeobs(G)\le 2n-4$.

Recently, Biedl and Mehrabi~\cite{BiedlM-GD17} studied \emph{non-blocking grid obstacle representations} of graphs, consisting of the pair $(\varphi, S)$ as before in which $\varphi$ maps the vertices of the graph to points in the plane and $S$ is a set of obstacles, but the adjacency in the graph is represented by replacing straight-line segments with $L_1$ shortest paths in the plane. That is, for every $u,w\in V(G)$, $uw\in E(G)$ if and only if some $L_1$ shortest path from $\varphi(u)$ to $\varphi(w)$ is disjoint from $\cup S$; see Figure~\ref{fig:example2D} for an illustration of these obstacle representations.

\begin{figure}[t]
\centering
\includegraphics[width=1.00\textwidth]{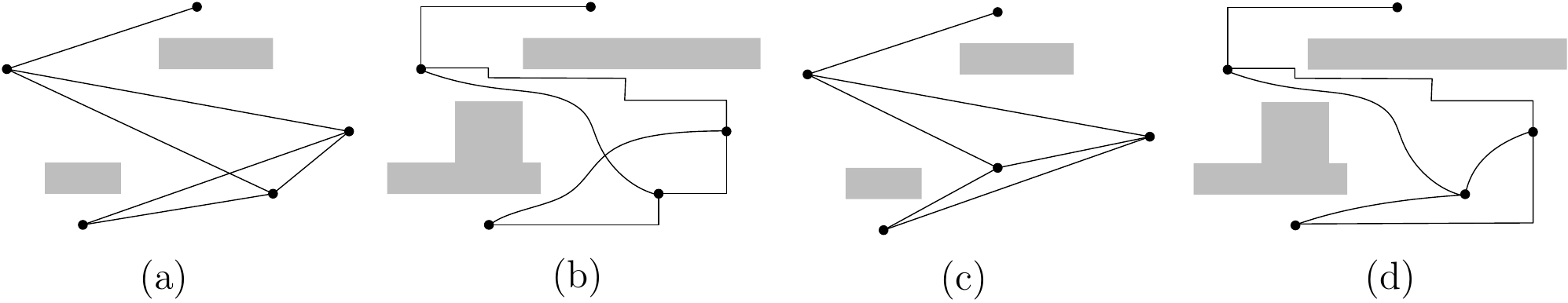}
\caption{Four different obstacle representations of the same graph $G$: (a) an obstacle representation, (b) a geodesic obstacle representation under $L_1$ distance, (c) a plane obstacle representation, and (d) a non-crossing geodesic obstacle representation under $L_1$ distance.}
\label{fig:example2D}
\end{figure}

\subparagraph{Geodesic obstacle representation.} In this paper, we generalize the notions of obstacle representations~\cite{AKL10}, plane obstacle representations, and grid obstacle representations~\cite{BiedlM-GD17} by introducing \emph{geodesic obstacle representations} of graphs. This natural generalization of obstacle representations comes from viewing line segments as shortest paths in the Euclidean plane. An obstacle representation $(\varphi,S)$ has the property that $uw\in E(G)$ if and only if the shortest path from $\varphi(u)$ to $\varphi(w)$ does not intersect $\cup S$. The Euclidean distance is a very special case because the shortest path between any two points $p$ and $q$ is unique. To accommodate other distance measures, we extend the definition of obstacle representation by saying that $uw\in E(G)$ if and only if \emph{some} shortest path from $\varphi(u)$ to $\varphi(w)$ does not intersect $\cup S$. In this way, we can obtain many generalizations of obstacle representations by changing the underlying distance measure. For example, with the $L_1$ distance measure, every $xy$-monotone path is a shortest path. Therefore, if $(\varphi,S)$ is an obstacle representation under $L_1$, then $uw\in E(G)$ if and only if there is some $xy$-monotone path from $u$ to $w$ that avoids $\cup S$. Analogous to plane obstacle representations, we can define \emph{non-crossing} geodesic obstacle representations in which $\varphi$ defines a plane embedding of graph $G$. Under the $L_1$ metric, this non-crossing version is equivalent to non-blocking grid obstacle representations as defined by Biedl and Mehrabi~\cite{BiedlM-GD17}.

Considering the $L_1$ metric in the plane, one can view a geodesic obstacle representation of $G$ as a partition of the neighbours of each vertex $u\in V(G)$ into four sets based on which of the four quadrants relative to $u$ the neighbours of $u$ are in the representation. Consequently, if $uv, vw\in E(G)$ in such a way that $v$ is in the same quadrant of $u$ as $w$ is in the quadrant of $v$ in a representation, then we must have $uw\in E(G)$ since there is an $xy$-monotone path from $u$ to $w$ in the representation. As such, it is now not clear whether every graph has a geodesic obstacle representation. Indeed, the focus of this paper is to determine, for a class $\mathcal{G}$ of graphs, whether or not every member of $\mathcal{G}$ has a geodesic obstacle representation (under some metric space). Clearly, the existence of such representations is more likely if one extends the definition of monotonicity by considering $2k$ equal-angled \emph{cones} around each vertex (instead of $2k=4$ quadrants), where $k>2$ is an integer. This leads us to the general question of, informally speaking, what is the minimum integer $k>0$ for which every member of $\mathcal{G}$ has a geodesic obstacle representation when shortest paths are defined by monotone paths relative to such $2k$ equal-angled cones around each vertex. In this paper, with this ``parameter $k$'', we study geodesic obstacle representations and its non-crossing version under polyhedral distance functions in $\R^d$ as well as shortest path distances in graphs. See Section~\ref{sec:prelimins} for the formal definition of this generalized notion of obstacle representations.

\subparagraph{Related work.} It is known that every $n$-vertex graph has obstacle number $O(n\log n)$~\cite{BalkoCV18} and some $n$-vertex graphs have obstacle number $\Omega(n/(\log\log n)^2)$ \cite{DujmovicM15}. For planar graphs, there exists planar graphs with obstacle number 2 (the icosahedron is an example~\cite{BermanCFGHW17}), but the best upper bound on the obstacle number of an $n$-vertex planar graph is $O(n)$. Recall the $O(n)$ upper bound on the plane obstacle number of any $n$-vertex planar graph by Euler's formula. A lower bound of $\Omega(n)$ is also not difficult: any plane drawing of the $\sqrt{n}\times\sqrt{n}$ grid $G_{\sqrt{n}\times\sqrt{n}}$ has at least $n-2\sqrt{n}$ bounded faces. Each of these faces has at least four vertices and therefore requires at least one obstacle, so $\planeobs(G_{\sqrt{n}\times\sqrt{n}})\ge n-2\sqrt{n}$. Gimbel et al.~\cite{GimbelMV17} have nailed the leading constant by showing that every planar graph has plane obstacle number at most $n-3$, the maximum being attained by planar bipartite graphs. See~\cite{AKL10,BalkoCV18,DujmovicM15,GimbelMV17} and the references therein for more details of results on obstacle number and its plane version.

While the obstacle numbers have been extensively studied under the Euclidean distance as shortest path, not much is known about obstacle representations under other shortest path metrics. In fact, we are only aware of the works of Bishnu et al.~\cite{BGMMP}, Biedl and Mehrabi~\cite{BiedlM-GD17} both of which considered only a restricted version of obstacle representations. Bishnu et al.~\cite{BGMMP} showed that any $n$-vertex planar graph $G$ has an obstacle representation on an $O(n^4)\times O(n^4)$ grid in the plane under $L_1$ metric, with the additional restriction that, for any $uw\in E(G)$, the shortest path from $\varphi(u)$ to $\varphi(w)$ also avoids $\varphi(v)$ for all $v\in V(G)\setminus\{u,w\}$ (in addition to avoiding $\cup S$). Biedl and Mehrabi~\cite{BiedlM-GD17} relaxed this ``vertex blocking'' constraint and were able to show that every $n$-vertex planar bipartite graph has a non-blocking grid obstacle representation on an $O(n)\times O(n)$ grid. They left open the problem of finding other classes of graphs for which such non-blocking grid obstacle representations exist and, in particular, whether every planar graph has such a representation.

\subparagraph{Our results.} In this paper, we prove the following results:
\begin{itemize}
\item For any integer $k>1$, there is a graph with $O(k^2)$ vertices that does not have a geodesic obstacle representation with parameter $k$. On the other hand, every $n$-vertex graph has a geodesic obstacle representation with every $k\ge n$.
\item For any integer $d>1$ and any integer $k>1$, there exists a graph that does not have a geodesic obstacle representation in $\R^d$ with parameter $k$. On the other hand, every $n$-vertex graph has a geodesic obstacle representation in $\R^3$ with $k=\lceil(1/2)\log_2 n+2\rceil$.
\item A planar graph $G$ has a non-crossing geodesic obstacle representation with $k=1$ if and only if $G$ is bipartite.
\item Every planar graph of treewidth at most 2 (and hence every outerplanar graph) has a non-crossing geodesic obstacle representation with $k=2$; i.e., a non-blocking obstacle representation.
\item Not every planar 3-tree has a non-crossing geodesic obstacle representation with $k=2$, answering the question asked by Biedl and Mehrabi~\cite{BiedlM-GD17} negatively. Moreover, not every planar 4-connected triangulation has a non-crossing geodesic obstacle representation with $k=2$.
\item Every planar 3-tree has a non-crossing geodesic obstacle representation with $k=3$. Furthermore, every 3-connected cubic planar graph has a non-crossing geodesic obstacle representation with $k=7$.
\item Every planar graph with maximum degree $\Delta$ has a non-crossing geodesic obstacle representation with $k=f(\Delta)$, where $f$ is a computable function depending only on $\Delta$.
\item Every $n$-vertex graph admits a non-crossing geodesic obstacle representation when taking the $D$-cube graph as the underlying distance metric, where $D=C\log n$ for some constant $C>0$.
\end{itemize}

\subparagraph{Organization.} We give definitions and notation in Section~\ref{sec:prelimins}. Then, we show our results for (general) geodesic obstacle representations in Section~\ref{sec:general} and for its non-crossing version in Section~\ref{sec:nonCrossing}. We give our results for graph metrics in Section~\ref{sec:graphMetric}, and conclude the paper with a discussion on open problems in Section~\ref{sec:conclusion}.

\section{Notation and Preliminaries}
\label{sec:prelimins}
Let $(X,\delta)$ be a metric space. A \emph{curve} over $X$ is a function $f:[0,1]\to X$.  We call $f(0)$ and $f(1)$ the \emph{endpoints} of the curve $f$ and define the \emph{image} of $f$ as $I(f)=\{f(t):0\le t\le 1\}$.  A curve $f$ is a \emph{geodesic} if, for every $0\le t\le 1$, $\delta(f(0),f(t))+\delta(f(t),f(1))=\delta(f(0),f(1))$. A \emph{path space} is a triple $(X,\delta,\mathcal{C})$, where $(X,\delta)$ is a metric space and $\mathcal{C}$ is a set of curves over $X$ that has the following closure property: if the curve $f$ is in $\mathcal{C}$ then, for every $0\le t\le 1$, $\mathcal{C}$ also contains the curves $g(x)=f(x\cdot t)$ and $h(x)=f(t+x\cdot (1-t))$. A path space $(X,\delta,\mathcal{C})$ is \emph{connected} if, for every distinct pair $u,w\in X$, there is some path in $\mathcal{C}$ with endpoints $u$ and $w$. For a path space $P=(X,\delta,\mathcal{C})$ and a subset $R\subset X$, we denote the subspace induced by $R$ as $P[R]=(R,\delta,\{f\in \mathcal{C}:I(f)\subseteq R\})$. The subspace that \emph{avoids} $R$ is defined as $P\setminus R= P[X\setminus R]$. Moreover, any curve in $P\setminus R$ is called an \emph{$R$-avoiding curve}. With these definitions in hand, we are ready to define a generalization of obstacle representations.
\begin{definition}
An \emph{$(X,\delta,\mathcal{C})$-obstacle representation} of a graph $G$ is a pair $(\varphi, S)$ where $\varphi:V(G)\to X$ is a one-to-one mapping and $S$ is a set of connected subspaces of $(X,\delta,\mathcal{C})$ with the property that, for every $u,w\in V(G)$, $uw\in E(G)$ if and only if $\mathcal{C}$ contains a $\cup S$-avoiding geodesic with endpoints $\varphi(u)$ and $\varphi(w)$.
\end{definition}

Notice that it is now not clear whether every graph has an $(X,\delta,\mathcal{C})$-obstacle representation. Indeed, the focus of this paper is to determine, for a class $\mathcal{G}$ of graphs and a particular path space $(X,\delta,\mathcal{C})$, whether or not every member of $\mathcal{G}$ has an $(X,\delta,\mathcal{C})$-obstacle representation. This is closely related to certain types of embeddings of $G$ into $X$. An \emph{embedding} $(\varphi,c)$ of a graph $G$ into $(X,\delta,\mathcal{C})$ consists of a one-to-one mapping $\varphi:V(G)\to X$ and a function $c: E(G)\to\mathcal{C}$ such that, such that for each $uw\in E(G)$, the endpoints of $c(uw)$ correspond to $\varphi(u)$ and $\varphi(w)$. The embedding is \emph{geodesic} if $c(uw)$ is a geodesic for every $uw\in E(G)$. Moreover, the embedding $(\varphi,c)$ is \emph{non-crossing} if $c(uw)$ is disjoint from $c(xz)$, for every $uw,xz\in E(G)$ with $\{u,w\}\cap \{x,z\}=\emptyset$. Observe that given an $(X,\delta,\mathcal{C})$-obstacle representation $(\varphi,S)$ of $G$, for each $uw\in E(G)$, we can choose some $\cup S$-avoiding geodesic $c(uw)\in \mathcal{C}$ with endpoints $\varphi(u)$ and $\varphi(w)$. Then, the pair $(\varphi,c)$ gives a geodesic embedding of $G$ into $X$. If we can choose $c$ such that $(\varphi,c)$ is also non-crossing, then we say that the representation $(\varphi,S)$ is \emph{non-crossing}.

\subparagraph{Distance functions.} In this paper, we focus on the obstacle representations using polyhedral distance functions in $\R^d$. For a set $N=\{v_0,\ldots,v_{t-1}\}$ of vectors in $\R^d$, we define the \emph{polyhedral distance function}
\[
\delta_N(p,q)=\min\left\{\sum_{i=0}^{t-1}|a_i|: \pi_S(q-p)=\sum_{i=0}^{t-1}a_iv_i\right\},
\]
where $S$ is the subspace spanned by the vectors in $N$ and $\pi_S(v)$ is the projection of $v$ on $S$. Every such distance function defines a centrally symmetric polyhedron $P_N=\{x\in\R^d: \delta_N(\mathbf{0},x)\le 1\}$. The facets of $P_N$ determine the geodesics. For a (closed) facet $F$ of $P_N$, we denote the \emph{cone} $C_F$ as the union of all rays originating at the origin and containing a point on $F$ (this is the affine hull of $F$).  For a point $x\in\R^d$, the \emph{$F$-sector} of $x$ is $Q^N_F(x)=C_F+x$. For a facet of $F$ of $P_N$, we say that a curve $f$ is \emph{$\delta_N$-monotone in direction $F$} if, for all $0\le a\le b\le 1$, $f(b)\in Q^N_F(f(a))$.  We say that a curve is $\delta_N$-\emph{monotone} if it is $\delta_N$-monotone in direction $F$ for some facet $F$ of $P_N$.  Observe that a curve $f$ is a geodesic for $\delta_N$ if and only if $f$ is $\delta_N$-monotone.
\begin{observation}
\label{obs:quadrilateral}
If $uw$ and $xz$ are curves that are each $\delta_N$-monotone in direction $F$ and $uw\cap xz$ contains at least one point $p$, then $\delta_k(u,z)=\delta_k(u,p)+\delta_k(p,z)$ and $\delta_k(x,w)=\delta_k(x,p)+\delta_k(p,w)$.
\end{observation}

When $X=\R^d$, we let $\mathcal{C}_d$ to denote the set of curves over $\R^d$. For the sake of compactness, when $X=\R^d$, we denote the $(\R^d,\delta_N,\mathcal{C}_d)$-obstacle representation by \emph{$\delta_N$-obstacle representation}. For the plane case $d=2$, we define, for each integer $k\ge 2 \in\mathbb{N}$, the \emph{regular distance function} $\delta_k=\delta_{N_k}$, where $N_k=\{(\cos(i\pi/k),\sin(i\pi/k)): i\in\{0,\ldots,2k-1\}\}$. In this case, the associated polygon $P_N$ is a regular $2k$-gon. Moreover, we use \emph{$\delta_k$-obstacle representation} as shorthand for $(\R^2,\delta_k,\mathcal{C}_2)$-obstacle representation. Moreover, for a point in $\R^2$, we denote the $i$-sector of $x$ by $Q^k_i(x)$, for $i\in\{0,\dots,2k-1\}$.

In addition to polyhedral distance functions, we consider obstacle representations under \emph{graph distance}. For a graph $H$, we denote the set of neighbours of a vertex $u$ in $H$ by $N_H(u)$ and the degree of $u$ by $\deg_H(u)$. Moreover, let $\delta_H$ denote the graph distance and let $\mathcal{C}_H$ be the set of curves that define paths in $H$. Then, we call a $(H,\delta_H,\mathcal{C}_H)$-obstacle representation an $H$-obstacle representation. If we consider the infinite square grid $H_4$ (resp., the infinite triangular grid $H_6$), for instance, then it is not difficult to argue that a graph $G$ has a non-crossing $\delta_2$-obstacle representation (resp., non-crossing $\delta_3$-obstacle representation) if and only if $G$ has a non-crossing $H_4$-obstacle representation (resp., non-crossing $H_6$-obstacle representation). In general, for any integer $D>1$, define the $D$-cube graph $Q_D$ to be the graph with vertex set $V(Q_D)=\{0,1\}^D$ and that contains the edge $uw$ if and only if $u$ and $w$ differ in exactly one coordinate.

\section{General Representations}
\label{sec:general}
In this section, we show our results for the general representations. We first consider the special case of $\R^2$ and will then discuss our results for higher dimensions. We start by the following result.
\begin{theorem}
\label{thm:plane-lower-bound}
For any $\epsilon >0$, there exists a graph $G$ with $n=n(\epsilon)$ vertices such that $G$ has no $\delta_k$-obstacle representation for any $k < n^{1-\epsilon}$.
\end{theorem}
\begin{proof}
For some constant $c>0$ and all sufficiently large $n$, there exists a graph $G$ with $n$ vertices and $cn^{2-2/r}$ edges and that contains no $K_{r,r}$ as subgraph~\cite{AlonKS03}. Let $(\varphi,S)$ be a $\delta_k$-obstacle representation of $G$ and let $(\varphi,c)$ be an embedding of $G$ obtained by taking, for each $uw\in E(G)$, $c(uw)$ to be some shortest $\cup S$-avoiding path from $\varphi(u)$ to $\varphi(w)$. From this point on we identify the vertices of $G$ with the points they are embedded to and the edges of $G$ with the curves the are embedded to.

By definition each edge $uw\in E(G)$ is $k$-monotone.  Since $P_V$ has at most $2k$ facets and each edge is monotone in at least two of these directions, this means that it has some facet $F$ such that $G$ contains $E(G)/k$ edges that are monotone in direction $F$.  Consider the graph $G'$ consisting of only these edges and the embedding $\varphi$ of $G'$. Observe that if two edges $uw$ and $xy$ of $G'$ intersect at some point $p$, then (after appropriate relabelling), this implies that there is a $\cup S$-avoiding geodesic from $u$ to $x$ as well as from $w$ to $y$. Therefore, $ux,uw\in E(G')$.

Therefore, if $G'$ contains an $r$-tuple of pairwise crossing edges, then $G'$ contains a $K_{r,r}$ subgraph.  Now, observe that the edges of $G'$ are monotone in some direction and (after an appropriate rotation) we can assume that they are $x$-monotone. We call this an \emph{$x$-monotone embedding}. Valtr~\cite{Valtr97} has shown that for every fixed $r$, there exists a constant $C=C(r)$ such that any $x$-monotone embedding of any $n$-vertex graph with more than $Cn\log n$ edges contains a a set of $r$ pairwise crossing edges. In our case, this means that $G$ contains a $K_{r,r}$ subgraph if $(cn^{2-2/r})/k\ge Cn\log n$, which gives a contradiction when $k\le cn^{1-2/r}/C\log n$. The result then follows by choosing any $r>2/\epsilon$.
\end{proof}

As $k\to\infty$, $\delta_k$ becomes the usual Euclidean distance function and $\delta_k$-obstacle representations are just the usual obstacle representations, which we know every graph has. Thus, for every $n\in \mathbb{N}$, there is a threshold value $k(n)$ such that every $n$-vertex graph has a $\delta_{k(n)}$-obstacle representation. Theorem~\ref{thm:plane-lower-bound} shows that $k(n)\in \Omega(n^{1-\epsilon})$ and the following theorem shows that $k(n)\in O(n)$.
\begin{theorem}
\label{thm:generalOnD}
Every $n$-vertex graph $G$ has a $\delta_k$-obstacle representation for $k=\lceil n/2\rceil$.
\end{theorem}
\begin{proof}
Consider the regular $2k$-gon with vertices at $(\cos((i+1/2)\pi/k), \sin((i+1/2)\pi/k))$, for $i\in\{0,\ldots,2k-1\}$.  It is well known that the pairs of vertices of this $2k$-gon determine only $k$ distinct directions and that these directions are $(i+1/2)\pi/k$ for $i=\{0,\ldots,k-1\}$.

Therefore, to obtain an obstacle representation of $G$, place its vertices on this regular $2k$-gon, join its vertices with straight-line segments and take the obstacles to be the faces of the resulting arrangement of line segments. That this is a $\delta_k$-obstacle representation follows from the fact that no $k$-monotone path uses two different directions determined by pairs of vertices and no two edges in the same direction cross each other.
\end{proof}

\subparagraph{Higher dimensions.} The proof of Theorem~\ref{thm:plane-lower-bound} makes critical use of the fact that obstacle representations live in the plane so
that any sufficiently dense (sub)graph has a $k$-tuple of pairwise crossing edges. An obvious question, then, is whether every graph has a $\delta_N$-obstacle representation in $\R^3$ (i.e., an $(\R^3,\delta_V, \mathcal{C}_3)$-obstacle representation), where $\delta_N$ is some polyhedral distance function. The following theorem shows that the answer to this question is no.

\begin{theorem}
\label{thm:rd-lower-bound}
Let $\delta_N$ be a polyhedral distance function over $\R^d$ whose corresponding polyhedron $P_N$ has $2k$ facets, for $k\in o(\log n)$. Then, there exists an $n$-vertex graph $G$ that has no $\delta_N$-obstacle representation.
\end{theorem}
\begin{proof}
Let $G$ be an $n$-vertex graph with no clique and no independent set of size larger than $2\log n$. The existence of such graphs was shown by Erd\H{o}s and Renyi~\cite{ErdosR59}. Suppose, for the sake of contradiction, that $G$ has some $\delta_{N}$-obstacle representation $(\varphi,S)$. Let $\prec$ denote lexicographic order over points in $\R^{d}$.

We will $k$-colour the $\binom{n}{2}$ pairs of vertices of $G$ where the colours are facets of $P_N$. A pair $(u,w)$ with $u\prec w$ is coloured with a facet $F$ of $P_N$ such that $w\in Q^N_F(u)$. If more than one such facet exists, we choose one arbitrarily. For each $i\in\{1,\ldots,k\}$, let $\prec_i$ denote the partial order obtained by restricting the total order $\prec$ to the pairs of vertices in $G$ with colour $i$.  We claim that for at least one $i$, $\prec_i$ contains a chain $v_1\prec_i\cdots\prec_i v_r$ of size $r\ge n^{1/k}$.  To see why this is so, observe that, by Dilwerth's Theorem, if $\prec_k$ does not contain a chain of length $n^{1/k}$, then it contains an antichain $A_k$ of size $n^{1-1/k}$.  Now, proceed inductively on $\prec_1,\ldots,\prec_{k-1}$ and $A_k$, observing that every pair in $A_k$ is coloured with $\{1,\ldots,k-1\}$.

Next, consider the relation $\prec_i'$ over $v_1,\ldots,v_r$ in which $v_a\prec_i v_b$ if and only if $1\le a<b\le r$ and $v_iv_j\in E(G)$. Observe that $\prec_i'$ is a partial order over $\{v_1,\ldots,v_r\}$. Therefore, by Dilwerth's Theorem, it contains a chain of size at least $\sqrt{r}$ or it contains an antichain of size at least $\sqrt{r}$. A chain corresponds to a clique in $G$ and an antichain corresponds to an independent set in $G$.  This contradicts our choice of $G$ when $\sqrt{r} > 2\log n$, which is true for all $k\in o(\log n)$ and all sufficiently large $n$.
\end{proof}

Theorem~\ref{thm:rd-lower-bound} shows that, for some $n$-vertex graphs $G$, any $\delta_{N}$-obstacle representation of $G$ must use a distance function $\delta_{N}$ with $k=\Omega(\log n)$ facets. Our next result shows that, even in $\R^3$, a polyhedral distance function with $k=O(\log n)$ facets is indeed sufficient.

\begin{theorem}
\label{thm:3d-universal}
Let $\delta_{N}$ be any polyhedral distance function in $\R^d$ for which the polyhedron $P_{N}$ has at least $2\log_2 n$ facets. Then, every $n$-vertex graph $G$ has a $\delta_{N}$-obstacle representation.
\end{theorem}
\begin{proof}
We claim that there exists a general position point set $X\subset\R^d$ of size at least $n$ with the property that no geodesic contains 3 points of $X$.  Given such a point set, we obtain an embedding $(\varphi,c)$  of $G$ by letting $\varphi$ be any one-to-one mapping of $V(G)$ onto $S$ and letting, for each $uw\in E(G)$, $c(uw)$ be the line segment with endpoints $u$ and $w$.  In this way, no path of length 2 or more in $G$ becomes a geodesic, so $(\varphi,c)$ is a $\delta_{N}$-obstacle representation of $G$. Furthermore, since $X$ is in general position, no two edges of the embedding cross. Therefore, taking $S=\R^d\setminus \bigcup_{uw\in E(G)} c(uw)$ yields a $\delta_{N}$-obstacle representation $(\varphi,S)$ of $G$. All that remains is to show the existence of the set $X$.  In the following, we will ignore the general position requirement on $X$, since it will be clear that the set $X$ we find can be slightly perturbed to ensure it is in general position.

Since $P_N$ is symmetric, the facets of $P_N$ come in $k\ge \log_2 n$ opposing pairs; let $\{f_1,\ldots,f_{k}\}$ contain one representative facet from each such pair, and let
   $\{c_1,\ldots,c_{k}\}$ be a set of balls, where each ball $c_i$ is
   contained in the interior of $f_i$. Finally, let $\{L_1,\ldots,L_{k}\}$
   be a set of sets of lines, where each $L_i$ contains all lines
   through the origin that intersect $c_i$.  Note that, since the
   balls $c_1,\ldots,c_{k}$ are disjoint, so are the line sets
   $L_1,\ldots,L_{k}$.

It suffices to construct a point set $X$, $|X|=2^{k}$, such that,
   for any triple $u,x,w\in X$, there exists $i,j\in\{1,\ldots,k\}$,
   $i\neq j$, such that $ux$ is parallel to some line in $L_i$ and $xw$ is
   parallel to some line in $L_j$.  We construct such a set inductively.
   If $k=0$, $X=\{\textbf{0}\}$ satisfies our requirements.

For $k\ge 2$, apply induction to obtain a set of points $X'$,
   $|X'|=2^{k-1}$ such that, for any triple $u,x,w\in X'$, there exists
   $i,j\in\{2,\ldots,k\}$, $i\neq j$, such that $ux$ is parallel to
   some line in $L_i$ and $xw$ is parallel to some line in $L_j$.  Now,
   choose two balls $A$ and $B$ such that, for every pair of points
   $u\in A$, $w\in B$, $u-w$ is parallel to some line in $L_{1}$.
   Finally, scale and translate $X'$ to obtain point set $X'_A\subset
   A$ and another point set $X'_B\subset B$ and take $X=X'_A\cup X'_B$.
   Clearly, $|X|=2^k$.  By the inductive hypothesis, if $\{u,x,w\}\in X'_A$
   or $\{u,x,w\}\in X'_B$, then $ux$ is parallel to some line $L_i$ and
   $xw$ is parallel to some line in $L_j$, with $i,j\in\{2,\ldots,k\}$,
   $i\neq j$.  Otherwise, assume without loss of generality that $u\in
   X'_A$ and $x,w\in X'_B$. Then, $u-x$ is parallel to some line in $L_1$
   and $x-w$ is parallel to some line $L_j$ for some $j\in\{2,\ldots,k\}$.
\end{proof}

If we take $t$ generic unit vectors in $\R^3$, then the polyhedral
distance function determined by these vectors defines a polyhedron
having $2t$ vertices and $4t-8$ triangular faces. Theorem~\ref{thm:3d-universal}
therefore implies that a polyhedral distance function determined by
$t\ge (1/2)\log_2 n +2$ unit vectors is sufficient to allow an obstacle
representation of any $n$-vertex graph.

In constant dimensions $d>3$, there exists sets of $t$ vectors in $\R^d$ defining polytopes with $\Theta(t^{\lfloor d/2\rfloor})$ facets. Therefore, in $\R^d$, every $n$-vertex graph has a $\delta_N$-obstacle representation with $|N|\in O(\sqrt[\lfloor d/2\rfloor]{\log n})$ vectors.

\section{Non-Crossing Representations}
\label{sec:nonCrossing}
In this section, we consider non-crossing $\delta_k$-obstacle representations. The following lemma shows that these representations are equivalent to plane $\delta_k$-obstacle embeddings.
\begin{lemma}
\label{lem:monotone}
A graph $G$ has a non-crossing $\delta_k$-obstacle representation if and only if $G$ has a non-crossing $\delta_k$-obstacle embedding.
\end{lemma}
\begin{proof}
First, suppose $G$ has a plane $\delta_k$-obstacle embedding $(\varphi,c)$. Then, we claim that by taking $S= \R^2\setminus \bigcup_{uw\in E(G)} c(uw)\cup \bigcup_{u\in V(G)}\varphi(u)$, we obtain an obstacle representation $(\varphi,S)$.

If $uw\in E(G)$ then $uw$ is an $S$-avoiding $k$-monotone curve.
   On the other hand if $uw\not\in E(G)$ then, since $(\varphi,c)$
   is plane, any $S$-avoiding $k$-monotone curve with endpoints $u$
   and $w$ would determine a $k$-monotone path from $u$ to $w$ in $G$.
   The definition of $\delta_k$-obstacle embedding does not allow this.

   Next, we argue that a non-crossing $\delta_k$-obstacle representation
   $(\varphi, S)$ implies the existence of a non-crossing $\delta_k$-obstacle
   embedding $(\varphi, c')$ of $G$.  of $G$.  By the
   definition of non-crossing, we immediately obtain a function
   $c:E(G)\to\mathcal{C}_2$ that maps edges of $G$ onto geodesic
   curves joining their endpoints, and any two of these curves are
   disjoint unless they share a common endpoint.

The resulting embedding $(\varphi,c)$ is almost a plane $\delta_k$-obstacle
   embedding except that it may contains pairs of edges $ux$ and $uz$
   such that $c(ux)$ and $c(uz)$ cross each other.  In this case,
   we observe that both $ux$ and $uz$ are $k$-monotone in direction
   $i$ for the same value of $i$.  This makes it easy to eliminate
   crossings by repeatedly swapping segments of the curves and making
   local modifications around the crossings.  Repeating this for every
   crossing pair of edges gives a plane $\delta_k$-obstacle embedding of $G$.
\end{proof}

Lemma~\ref{lem:monotone} allows us to focus our effort on studying the existence (or not) of plane $\delta_k$-obstacle embeddings.

\subsection{$\delta_1$-Obstacle Representations}

In this section, we show that the class of planar graphs that have plane
$\delta_1$-obstacle embeddings are exactly planar bipartite graphs. Since a
graph has a plane $\delta_1$-obstacle embedding if, and only if, it can be
straight-line embedded without $x$-monotone paths of length 2, we prove a
slightly more general result about such embeddings.

\begin{theorem}
  A combinatorial embedding of a planar graph (the counter-clockwise order of
  the neighbours of each vertex) has a straight-line plane embedding with the
  same neighbour orders and no $x$-monotone paths of length 2 if, and only if,
  the graph is bipartite.
\end{theorem}

\begin{proof}
  If the graph is not bipartite, then it has an odd cycle, whose embedding
  must have at least one $x$-monotone path of length 2. On the other hand, we
  show how to construct the desired straight-line plane embedding if the graph
  is bipartite and has at least 3 vertices. The construction has three stages:
  input transformations to simplify the embedding; the embedding itself; and
  adaptation of the embedding to the original input.
  
  The first input transformation adds vertices and edges to the graph so a
  2-connected quadrilateralization results. The graph is first made connected
  by repeatedly adding edges between outer vertices of different connected
  components. To make it 2-connected, we need to deal with cut vertices. If
  $u$ is a cut vertex, let $v$ be a neighbour of $u$ whose next vertex $w$ in
  the counter-clockwise order around $u$ lies in a different 2-connected
  component than $v$ and add a path of length 2 between $v$ and $w$, as in
  Figure~\ref{figure:bipartite.route-cut-vertex}. This path addition merges
  the 2-connected components of $v$ and $w$, so eventually a 2-connected graph
  remains.
  
  \begin{figure}[h]
    \centering
    {\includegraphics{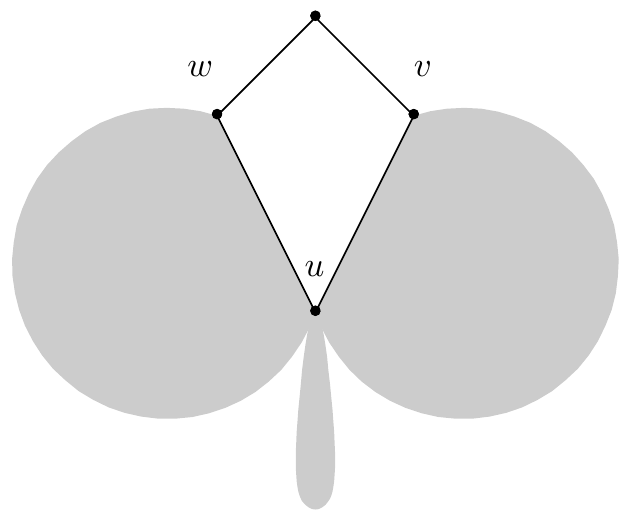}}
    \caption{\label{figure:bipartite.route-cut-vertex}Routing around a cut
    vertex.}
  \end{figure}
  
  Note that because the graph is 2-connected and has at least 3 vertices, a
  counter-clockwise traversal of any face will not repeat edges or vertices.
  Therefore, the neighbour ordering of the original graph is preserved if we
  preserve the faces of this 2-connected combinatorial embedding, and this is
  what the construction will achieve. To conclude this first input
  transformation, we obtain a quadrilateralization by repeatedly inserting
  edges between vertices three edges apart in faces with more than four
  vertices, which is possible since faces all have an even number of vertices
  greater than two.
  
  For the second input transformation, while there are two quadrilaterals
  sharing exactly two adjacent edges, we merge these quadrilaterals into one
  by erasing these two edges and the isolated vertex, as in Figure
  \ref{figure:bipartite.merging}. Note that no cut vertices are introduced.
  Also, this merging reduces the number of quadrilaterals by one, so this
  process eventually ends.
  
  \begin{figure}[h]
    \centering
    {\includegraphics{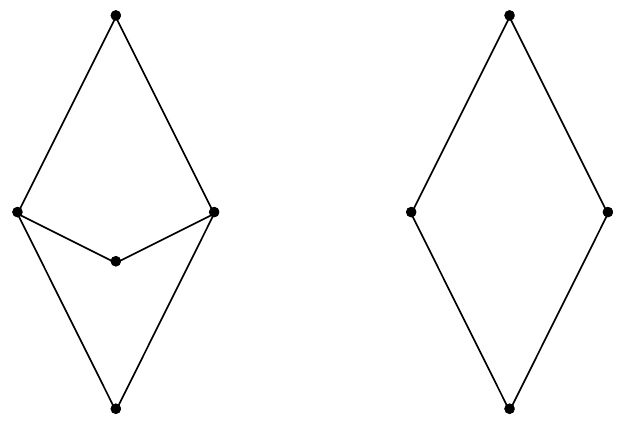}}
    \caption{\label{figure:bipartite.merging}Merging quadrilaterals to avoid
    this adjacency pattern.}
  \end{figure}
  
  Finally, for the third input transformation, while there is a vertex $v$
  that is not part of the outer face, we first name the neighbours of $v$ in
  counter-clockwise order as $w_0, \ldots, w_k$ and the vertex opposing $v$ in
  the quadrilateral to the left of $\overrightarrow{v w_k}$ as $u$. Then we
  remove $v$ and its incident edges and connect $u$ to $w_1, \ldots, w_{k -
  1}$, as in Figure \ref{figure:bipartite.removing-internal-vertex}. Note that
  no cut vertices are introduced here either and that the quadrilaterals
  originally incident to $v$ did not share any edges but the ones incident to
  $v$ due to our second input transformation. Furthermore, each of these steps
  removes a vertex, so this process also terminates.
  
  \begin{figure}[h]
    \centering
    {\includegraphics{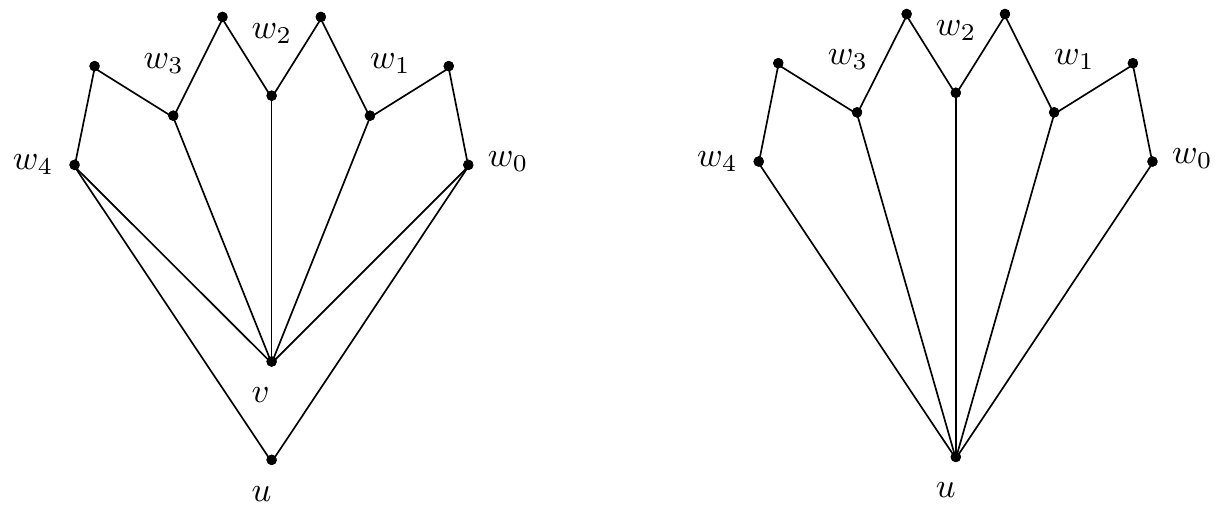}}
    \caption{\label{figure:bipartite.removing-internal-vertex}Removing an
    internal vertex.}
  \end{figure}
  
  The reason performing the embedding is now easy is that we are dealing with
  a 2-connected outerplanar graph, whose dual graph is thus a tree. We can
  thus remove leaves from this tree until a single vertex remains. Therefore,
  we can start by embedding the quadrilateral corresponding to this vertex in
  any feasible way and keep embedding quadrilaterals that share a single
  (outer) edge with the current embedding, one at a time. These new
  quadrilaterals are embedded as concave quadrilaterals small enough so that
  they do not overlap with the rest of the embedding and so that their edges
  have the same slope as the shared edge, as in Figure~\ref{figure:bipartite.leaf-embedding}. Because they have the same slope, no
  $x$-monotone paths of length 2 are created.
  
  \begin{figure}[h]
    \centering
    {\includegraphics{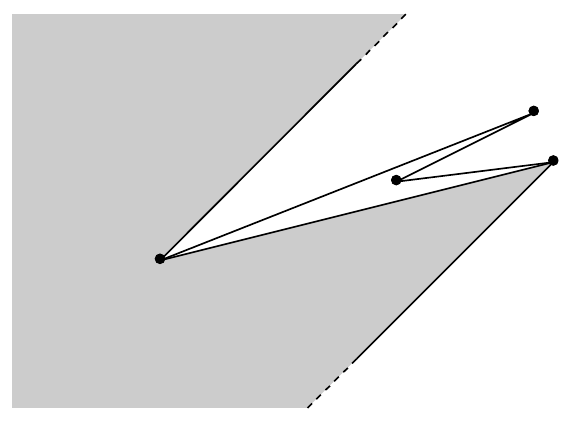}}
    \caption{\label{figure:bipartite.leaf-embedding}Attaching a new
    quadrilateral to an edge.}
  \end{figure}
  
  Now we have an embedding with no $x$-monotone paths of length 2, but we need
  to ``undo'' the transformations we did to the original graph in the reverse
  order they were made. To undo a transformation of the third type (Figure~\ref{figure:bipartite.removing-internal-vertex}), note that $w_0, \ldots,
  w_k$ must have been embedded all to the left or all to the right of $u$.
  W.l.o.g., we assume the latter case. First we erase the edges from $u$ to
  $w_1, \ldots, w_{k - 1}$ and we embed the vertex $v$ close to $u$, to the
  left of $\overrightarrow{u w_0}$ and to the right of $\overrightarrow{u
  w_k}$. If $v$ is close enough to $u$, there will be no crossings or change
  in slope when we reinsert the edges from $v$ to $w_0, \ldots, w_k$, as in
  Figure~\ref{figure:bipartite.undoing-internal-vertex-removal}.
  
  \begin{figure}[h]
    \centering
    {\includegraphics{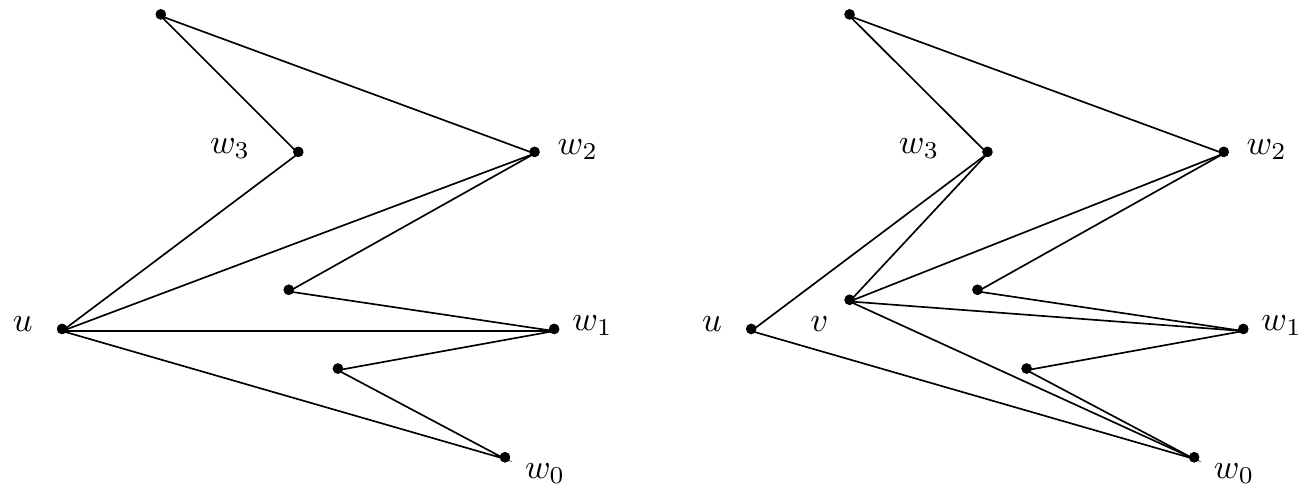}}
    \caption{\label{figure:bipartite.undoing-internal-vertex-removal}Undoing
    an internal vertex removal.}
  \end{figure}
  
  Upon close inspection, the second type of transformation (Figure~\ref{figure:bipartite.merging}) can be interpreted as a transformation of
  the third type where $k = 1$. The procedure just outlined may be then used
  to obtain a valid embedding prior to these transformations. As for
  transformations of the first type, they are simply vertex and edge
  insertions, so removing these vertices and edges clearly does not create
  $x$-monotone paths of length 2, providing us our desired embedding.
\end{proof}

\begin{corollary}
  A planar graph has a straight-line $\delta_1$-obstacle embedding if, and
  only if, it is bipartite.
\end{corollary}

\subsection{$\delta_2$-Obstacle Representations}
In this section, we focus on plane $\delta_2$-obstacle embeddings. Recall that these are equivalent to the non-blocking planar grid obstacle representation studied by Biedl and Mehrabi~\cite{BiedlM-GD17}. We begin with the positive result that all graphs of treewidth at most 2 (i.e., partial 2-trees) have plane $\delta_2$-obstacle embeddings.

\subparagraph{Treewidth.} A \emph{$k$-tree} is any graph that can be obtained in the following manner: we begin with a clique on $k+1$ vertices and then we repeatedly select a subset of the vertices that form a $k$-clique $K$ and add a new vertex adjacent to every element in $K$. The class of $k$-trees is exactly the set of edge-maximal graphs of treewidth $k$. A graph $G$ is called a \emph{partial $k$-tree} if it is a subgraph of some $k$-tree. The class of partial $k$-trees is exactly the class of graphs of treewidth at most $k$. We will make use of the following lemma, due to Dujmovi\'c and Wood~\cite{DujmovicW07} in proving Theorem~\ref{thm:2-tree} and later in Section~\ref{subsec:higherK}.
\begin{lemma}[Dujmovi\'c and Wood~\cite{DujmovicW07}]
\label{lem:dujwood}
Every $k$-tree is either a clique on $k+1$ vertices or it contains a non-empty independent set $S$ and a vertex $u\not\in S$, such that (i) $G\setminus S$ is a $k$-tree, (ii) $\deg_{G\setminus S}(u)=k$, and (iii) every element in $S$ is adjacent to $u$ and $k-1$ elements of $N_{G\setminus S}(u)$.
\end{lemma}

\begin{theorem}
\label{thm:2-tree}
Every partial 2-tree has a plane straight-line $\delta_2$-obstacle embedding.
\end{theorem}
\begin{proof}
Let $G$ be a partial 2-tree. We can, without loss of generality, assume that $G$ is connected.  If $|V(G)|< 4$, then the result is trivial, so we can assume $|V(G)|\ge 4$.  We now proceed by induction on $|V(G)|$.

Let $T=T(G)$ be a 2-tree with vertex set $V(G)$ and that contains $G$. Apply Lemma~\ref{lem:dujwood} to find the vertex set $S$ and the vertex $u$. Let $x$ and $y$ be the neighbours of $u$ in $T\setminus S$. Now, apply induction to find a plane straight-line $\delta_2$-obstacle embedding of the graph $G'$ whose vertex set is $V(G')=V(G)\setminus S$ and whose edge set is $E(G')=E(G\setminus S)\cup\{ux,uy\}$. Denote by $S_x$ (resp., $S_y$) the neighbours of $x$ (resp., $y$) that belong to $S$.

Now, observe that, since $u$ has degree 2 in $G'$ and the edges $ux$
  and $uy$ are in $G'$, this embedding does not contain any monotone path
  of the form $uxw$ or $uyw$ for any $w\in V(G)\setminus\{u,x,y\}$.
  Therefore, if we place the vertices in $S$ sufficiently close to $u$,
  we will not create any monotone path of the form $ayw$ or $axw$ for
  any $a\in S$ and any $w\in V(G)\setminus \{u,x,y\}$.  What remains
  is to show how to place the elements of $S$ in order to avoid unwanted
  monotone paths of the form $uay$, $uax$, or $aub$ for any $a,b\in S$.
  There are three cases to consider:

\begin{figure}[t]
\centering
\includegraphics[width=0.90\textwidth]{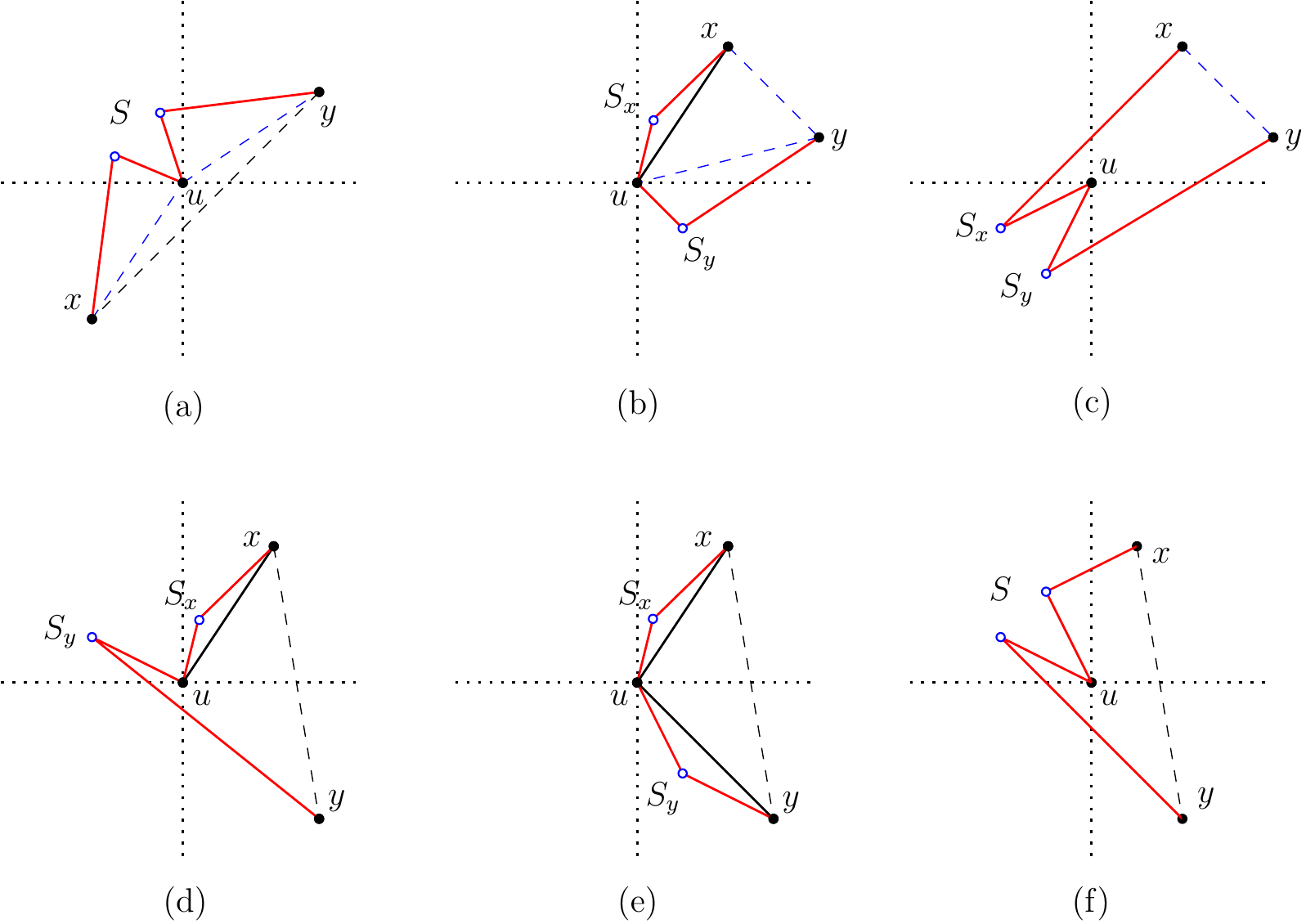}
\caption{An illustration in supporting the proof of Theorem~\ref{thm:2-tree}.}
\label{fig:2treesBoth}
\end{figure}

\begin{enumerate}
\item  $x\in Q^2_i(u)$ and $y\in Q^2_{i+2}(u)$ for some $i\in\{0,\ldots,3\}$. W.l.o.g., assume that $Q^2_{i+3}(u)$ does not intersect the segment $xy$. Then, we can embed the elements of $S$ in $Q^2_{i+3}(u)$ without creating any new monotone paths; see Figure~\ref{fig:2treesBoth}(a).
\item $x,y\in Q^2_i(u)$ for some $i\in\{0,\ldots,3\}$. There are two subcases: \begin{inparaenum}[(i)] \item At least one of $ux$ or $uy$ is in $E(G)$. Suppose $ux\in E(G)$. Then we embed $S_x$ in $Q^2_i(u)$ and embed $S_y$ in $Q^2_{i+3}(u)$; see Figure~\ref{fig:2treesBoth}(b). The only monotone paths this creates are of the form $uax$ with $a\in S_x$, which is acceptable since $ux\in E(G)$. \item Neither $ux$ nor $uy$ is in $E(G)$. In this case, we embed all of $S$ in $Q^2_{i+2}(u)$ (see Figure~\ref{fig:2treesBoth}(c)). This does not create any new monotone paths. \end{inparaenum}
\item $x\in Q^2_i(u)$ and $y\in Q^2_{i+3}(u)$ for some $i\in\{0,\ldots,3\}$. We have three subcases to consider: \begin{inparaenum}[(i)] \item $|\{ux,uy\}\cap E(G)|=1$. In this case, assume $ux\in E(G)$. Then, we embed the vertices of $S_x$ in $Q^2_i(u)$ and we embed the vertices of $S_y$ in $Q^2_{i+1}(u)$. See Figure~\ref{fig:2treesBoth}(d). The only monotone paths this creates are of the form $uax$ with $a\in S_x$, which is acceptable since $ux\in E(G)$. \item $|\{ux,uy\}\cap E(G)|=2$. In this case, we embed the vertices of $S_x$ in $Q^2_i(u)$ and we embed the vertices of $S_y$ in $Q^2_{i+3}(u)$ (see Figure~\ref{fig:2treesBoth}(e)). The only monotone paths this creates are of the form $uax$ with $a\in S_x$ and $uby$ with $b\in S_y$, which is acceptable since $ux,uy\in E(G)$. \item $|\{ux,uy\}\cap E(G)|=0$.  In this case, we embed all of $S$ into $Q^2_{i+1}$ (see Figure~\ref{fig:2treesBoth}(f)). This does not create any new monotone paths. \end{inparaenum}
\end{enumerate}
This completes the proof of the theorem.
\end{proof}

We next show that not every planar 3-tree admits a plane $\delta_2$-obstacle embedding. To this end, we first need some preliminary results.
\begin{lemma}
\label{lem:labelling}
The vertices of any triangle $xyz$ can be labelled such that $y,z\in Q^2_i(x)$ for some $i\in\{0,\ldots,3\}$.
\end{lemma}
\begin{proof}
Consider the vertex $x$ and assume w.l.o.g. that $y\in Q^2_i(x)$, for some $i\in\{0,\dots,3\}$. Notice that $x\in Q^2_{i+2}(y)$. If $z\in Q^2_i(x)$ or $z\in Q^2_{i+2}(y)$, then we are done. Otherwise, we must have $x,y\in Q^2_{i+1}(z)$ or $x,y\in Q^2_{i+3}(z)$, which proves the lemma by a re-labelling.
\end{proof}

A (1-level) \emph{subdivision} of a triangle $xyz$ is obtained by adding a vertex $w$ in the interior of $xyz$ and adding the edges $wx$, $wy$, $wz$. A $d$-level subdivision of $xyz$ is obtained by repeating this process recursively to a depth of $d$.
\begin{lemma}
\label{lem:level-1}
Let $G$ be a non-crossing $\delta_2$-obstacle embedding of some graph, and let $xyz$ be a three-cycle in $G$ embedded with $x\in Q^2_i(y)$ and $z\in Q^2_i(x)$. Then, $xyz$ does not contain a 3-level subdivision in its interior.
\end{lemma}
\begin{proof}
W.l.o.g., assume that $i=0$ and $x$ is above the edge $yz$. Consider the location of the vertex $w$ that subdivides $xyz$. There are three cases to consider:
\begin{enumerate}
      \item The vertex $w$ is placed in $Q^2_0(x)$.  In this case,
        there will be a $d_2$-monotone path from $z$ to the vertex $w'$ that
        subdivides $xyw$.
      \item The vertex $w$ is placed in $Q^2_2(x)$.  In this case,
        there will be a $d_2$-monotone path from $y$ to the vertex $w'$ that
        subdivides $xwz$.
      \item The vertex $w$ is placed in $Q^2_3(x)$. In this case,
      consider the vertex $w'$ that subdivides $zwy$.  The preceding
      two arguments prevent $w'$ from being placed in $Q^2_0(w)$
      or $Q^2_2(w)$.  However, placing $w'$ in $Q^2_3(w)$ creates a
      monotone path from $x$ to $w'$.
   \end{enumerate} 
\end{proof}

\begin{lemma}
\label{lem:level-2}
Let $G$ be a non-crossing $\delta_2$-obstacle embedding of some graph, and let $xyz$ be a three-cycle in $G$ with $yz\in Q^2_i(x)$ for some $i$. Then, $xyz$ does not contains a 4-level subdivision in its interior.
\end{lemma}
\begin{proof}
If $xyz$ does not already meet the criteria for Lemma~\ref{lem:level-1}, then any choice of location for the first-level subdivision vertex will create at least one triangle that does meet the criteria for Lemma~\ref{lem:level-1}.
\end{proof}

\begin{theorem}
\label{thm:stellated}
There exists a planar 3-tree that does not have a non-crossing $\delta_2$-obstacle embedding.
\end{theorem}
\begin{proof}
Consider the graph $G$ that is a 5-level subdivision of a triangle. In any embedding of $G$, there is a triangle $xyz$ with a 4-level subdivision in its interior. The theorem then follows since, by Lemma~\ref{lem:labelling}, we can apply Lemma~\ref{lem:level-1} to $xyz$.
\end{proof}

Notice that the graph in Theorem~\ref{thm:stellated} has treewidth 3. Figure~\ref{fig:triangular} shows that the infinite triangular grid has a non-crossing $\delta_2$-obstacle embedding. This means that while not all planar graphs of treewidth 3 have a non-crossing $\delta_2$-obstacle embedding, there are planar graphs of treewidth $\Theta(\sqrt{n})$ that admit such an embedding.

\begin{figure}[!h]
   \begin{center}
      \includegraphics[width=\textwidth]{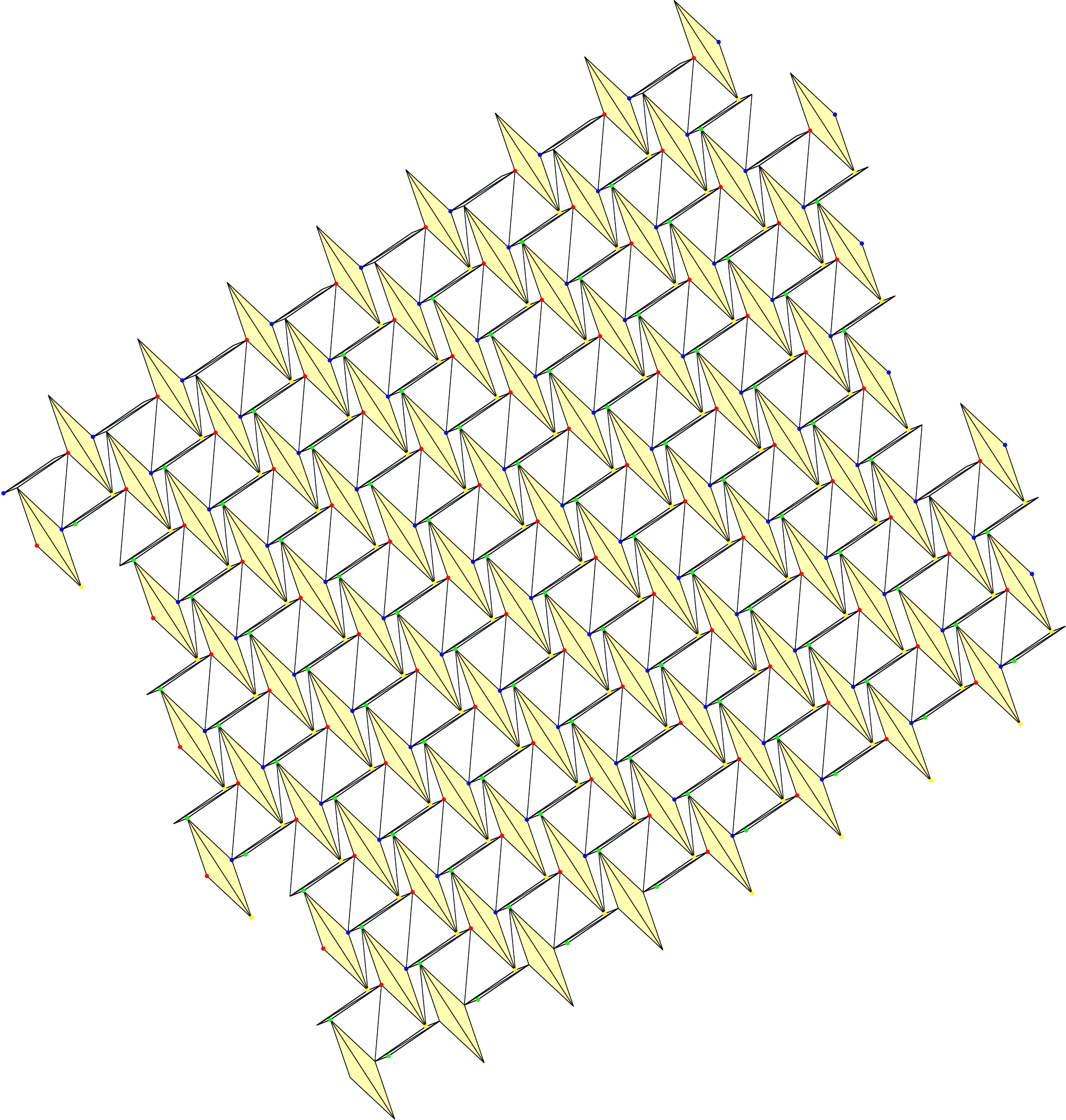}
   \end{center}
   \caption{A plane $\delta_2$-obstacle embedding of the triangular grid.}
   \label{fig:triangular}
\end{figure}

We next prove that even 4-connectivity does not help to guarantee the existence of non-crossing $\delta_2$-obstacle embeddings. The idea is to show that a 4-connected triangulation having a plane $\delta_2$-obstacle representation must have a constrained 4-colouring in the sense that, for the neighbours of a vertex, which colours and in what order are they allowed to be assigned to them. We then find a 4-connected triangulation that does not have such a constrained 4-colouring. We next give the details. Let $G$ be a non-crossing $d_2$-obstacle representation of a 4-connected triangulation; we call a vertex of $G$ an \emph{outer vertex} if it is incident to the outerface of $G$.
\begin{lemma}
\label{lem:fourTypesOfVertices}
For every internal vertex $u$ of $G$, there is an $i\in\{0,\dots,3\}$ such that $u$ has exactly one neighbour in $Q^2_{i-1}(u)$, exactly one neighbour in $Q^2_{i+1}(u)$ and all remaining neighbours in $Q^2_i(u)$.
\end{lemma}
\begin{proof}
Suppose for a contradiction that this is not the case. Then, since $G$ is 4-connected, $u$ would have two non-adjacent neighbours $x$ and $y$ such that $x\in Q^2_j(u)$ and $y\in Q^2_{j+2}(u)$, for some $j\in\{0,\dots,3\}$. This means that the path $xuy$ is $xy$-monotone, but $x$ and $y$ are not adjacent in $G$ --- a contradiction.
\end{proof}

Lemma~\ref{lem:fourTypesOfVertices}, classifies the internal vertices of $G$ into four types 0, 1, 2, and 3. For each internal vertex $u\in V(G)$, we define $c(u)$ as the \emph{type} of the vertex $u$. For an internal vertex $u$ with $c(u)=i$, we can assume by Lemma~\ref{lem:fourTypesOfVertices} that $u$ has $\deg(u)-2$ of its neighbours in $Q^2_i(u)$ and has no neighbours in $Q^2_{i+2}(u)$.

\begin{lemma}
\label{lem:typeOfTwoNeighbours}
For an internal vertex $u$ of $G$ with $c(u)=i$, for some $i\in\{0,\dots,3\}$, the neighbour $x$ (resp., $y$) of $u$ in $Q^2_{i+1}(u)$ (resp., $Q^2_{i-1}(u)$) is either an outer vertex or $c(x)=i-1$ (resp., $c(y)=i+1$). \end{lemma}
\begin{proof}
Since $x\in Q^2_{i+1}(u)$ and $y\in Q^2_{i-1}(u)$, the vertices $u$ and $y$ are two neighbours of $x$ that are both in $Q^2_{i-1}(x)$. This means that $c(x)=i-1$, unless $x$ is an outer vertex. An analogous argument applies to the vertex $y$.
\end{proof}

Lemma~\ref{lem:typeOfTwoNeighbours} suggests the following property for an internal vertex $u$ of $G$. If $c(u)=i$ and all neighbours of $u$ are internal (i.e., $u$ has no neighbour on the outerface), then $u$ has two consecutive neighbours $x$ and $y$ such that (i) $c(x)=i-1$ and $c(y)=i+1$, and (ii) $x$ is before $y$ in the counter-clockwise ordering of the neighbours of $u$. See Figure~\ref{fig:2neighbours}. We call this as the \emph{two-neighbour property} of $u$.

\begin{figure}[t]
\centering
\includegraphics[width=1.00\textwidth]{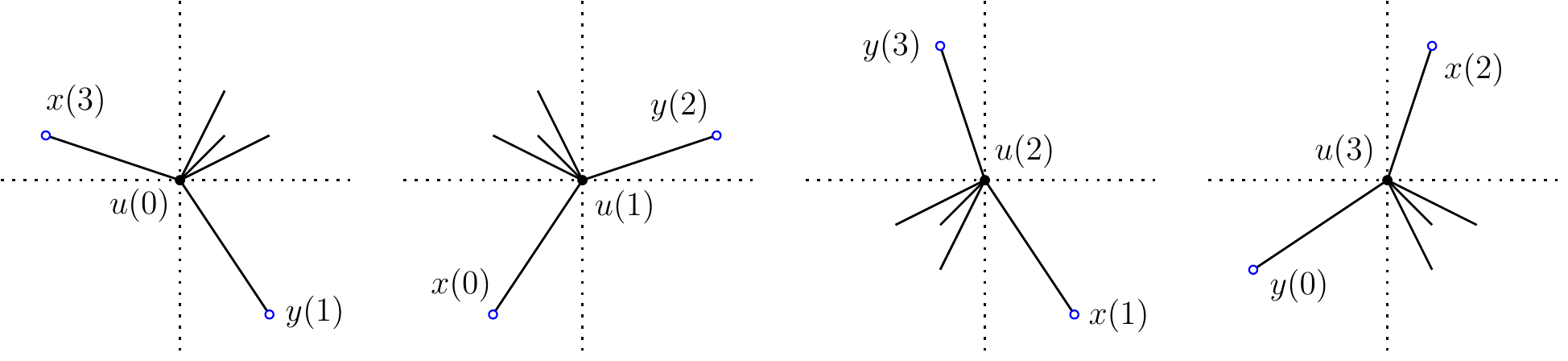}
\caption{The two neighbours $x$ and $y$ of $u$ in four cases depending on the type of $u$. The value besides each vertex denotes its type.}
\label{fig:2neighbours}
\end{figure}

\begin{lemma}
\label{lem:properColouring}
The partial function $c:V(G)\rightarrow\{0,1,2,3\}$ is a proper colouring of the internal vertices of $G$.
\end{lemma}
\begin{proof}
Let $u$ be an internal vertex $G$ with $c(u)=i$. Then, by Lemma~\ref{lem:typeOfTwoNeighbours}, the neighbour $x$ (resp., $y$) of $u$ in $Q^2_{i+1}(u)$ (resp., in $Q^2_{i-1}(u)$) is either an outer vertex or $c(x)=i-1$ (resp., $c(y)=i+1$). Moreover, any vertex $z\in Q^2_i(u)$ has at least one neighbour (namely, $u$) in $Q^2_{i+2}(z)$ and so $z$ is either an outer vertex or $c(z)\neq i$.
\end{proof}

Now, consider the graph $H$ shown in Figure~\ref{fig:4connectedExample}(a) in which the labels of the vertices denote the colour of the vertices.

\begin{lemma}
\label{lem:orderingOfHVertices}
Let $G$ be a planar graph that contains $H$ such that all the vertices of $H$ are internal in $G$. If $G$ has a plane $\delta_2$-obstacle embedding, then (referring to the graph $H$ shown in Figure~\ref{fig:4connectedExample}(a)) $c(e)=i$, $c(h)=i+1$, $c(g)=i+2$ and $c(f)=i+3$ up to rotation.
\end{lemma}
\begin{proof}
Since $G$ is planar and all the vertices of $H$ are internal in $G$, the two-neighbour property must hold for every vertex $u$ of $H$ by Lemma~\ref{lem:typeOfTwoNeighbours}. W.l.o.g., assume that $a=0$ and so $c\in\{1,2,3\}$ by Lemma~\ref{lem:properColouring}. We next consider these three cases.

\noindent{\bf Case I: $a=0$ and $c=1$.} Then, we consider the two cases for $b$. (i) If $b=3$, then $f\neq 1$ because otherwise $b$ cannot satisfy the two-neighbour property (i.e., $b$ cannot have two consecutive neighbours $x$ and $y$ such that $c(x)=2$, $c(y)=0$, and $x$ appears before $y$ in the counter-clockwise ordering around $b$). By Lemma~\ref{lem:properColouring}, $f=2$ and so we then have $g=0$. If $h=2$, then we must have $d=3$ in order to satisfy the two-neighbour property for $h$. This will then imply that $e=1$, but then the two-neighbour property does not hold for $d$. If $h=3$, then $d=2$ and $e=1$ by Lemma~\ref{lem:properColouring}. But then, $h$ does not have the two-neighbour property. (ii) If $b=2$, then we must have $g=3$ in order to satisfy the two-neighbour property for $b$. This implies that $f=1$ by proper colouring, but now the two-neighbour property does not hold for $g$.

\noindent{\bf Case II: $a=0$ and $c=2$.} We consider the two cases for $b$. (i) If $b=1$, then $g=0$. Notice that $g$ cannot be 3 because then $f=2$ by proper colouring, but then the two-neighbour property does not hold for $g$. Since $g=0$ and the colours 1 and 3 must appear at two consecutive neighbours of $c$ in counter-clockwise, we must have $d=1$ and $h=3$. This gives $e=2$ by proper colouring, but then the two-neighbour property does not hold for $e$ (notice that $d=1$ and $h=3$, but they do not appear in counter-clockwise around $e$). (ii) If $b=3$, then $g\in\{0,1\}$. If $g=0$, then the remaining vertices are forced to be $d=1$ and $h=3$ (by the two-neighbour property for $c$), and then $e=2$ and $f=1$ (by proper colouring) in this order. Similarly, if $g=1$, then $f=2$ (by the two-neighbour property for $b$), $e=3$ and $d=1$ (by the two-neighbour property for $a$), and $h=0$ by proper colouring. Observe that in either case $c(e)=i$, $c(h)=i+1$, $c(g)=i+2$ and $c(f)=i+3$ up to rotation, satisfying the ordering stated in the lemma.

\noindent{\bf Case III: $a=0$ and $c=3$.} We again consider the two cases for $b$. (i) If $b=1$, then we must have $g=0$ and $f=2$ (in order to satisfy the two-neighbour property for $b$). This will then imply that we have $d=2$ and $h=0$ (in order to satisfy the two-neighbour property for $c$). But, then we cannot satisfy this property for $h$. (ii) If $b=2$, then we must have $g=1$ and $f=3$ (again, in order to satisfy the two-neighbour property for $b$). But, then we cannot satisfy this property for $f$; notice that although $b=2$ and $a=0$, they do not appear in the counter-clockwise around $f$ as required.

Therefore, the only way $G$ can have a plane $\delta_2$-obstacle embedding is to have $a=0$ and $c=2$ in which case we get $c(e)=i$, $c(h)=i+1$, $c(g)=i+2$ and $c(f)=i+3$ up to rotation.
\end{proof}

\begin{theorem}
\label{thm:existsA4Connected}
There exists a 4-connected triangulation $G$ with maximum degree 7 that has no plane $\delta_2$-obstacle embedding.
\end{theorem}
\begin{proof}
Graph $G$ is shown in Figure~\ref{fig:4connectedExample}(b). First, it is easy to see that $G$ is planar, 4-connected and has maximum degree 7. Moreover, $G$ contains two copies of graph $H$ ``attached'' to each other in its interior. By Lemma~\ref{lem:orderingOfHVertices}, the four vertices on the outerface of each of these copies of $H$ must have the ordering type given in Lemma~\ref{lem:orderingOfHVertices}. However, since these two copies of $H$ share two of the outerface vertices, it is not possible to satisfy such ordering type for the outerface vertices at the same time. As such, $G$ does not admit a plane $\delta_2$-obstacle embedding.
\end{proof}

\begin{figure}[t]
\centering
\includegraphics[width=0.80\textwidth]{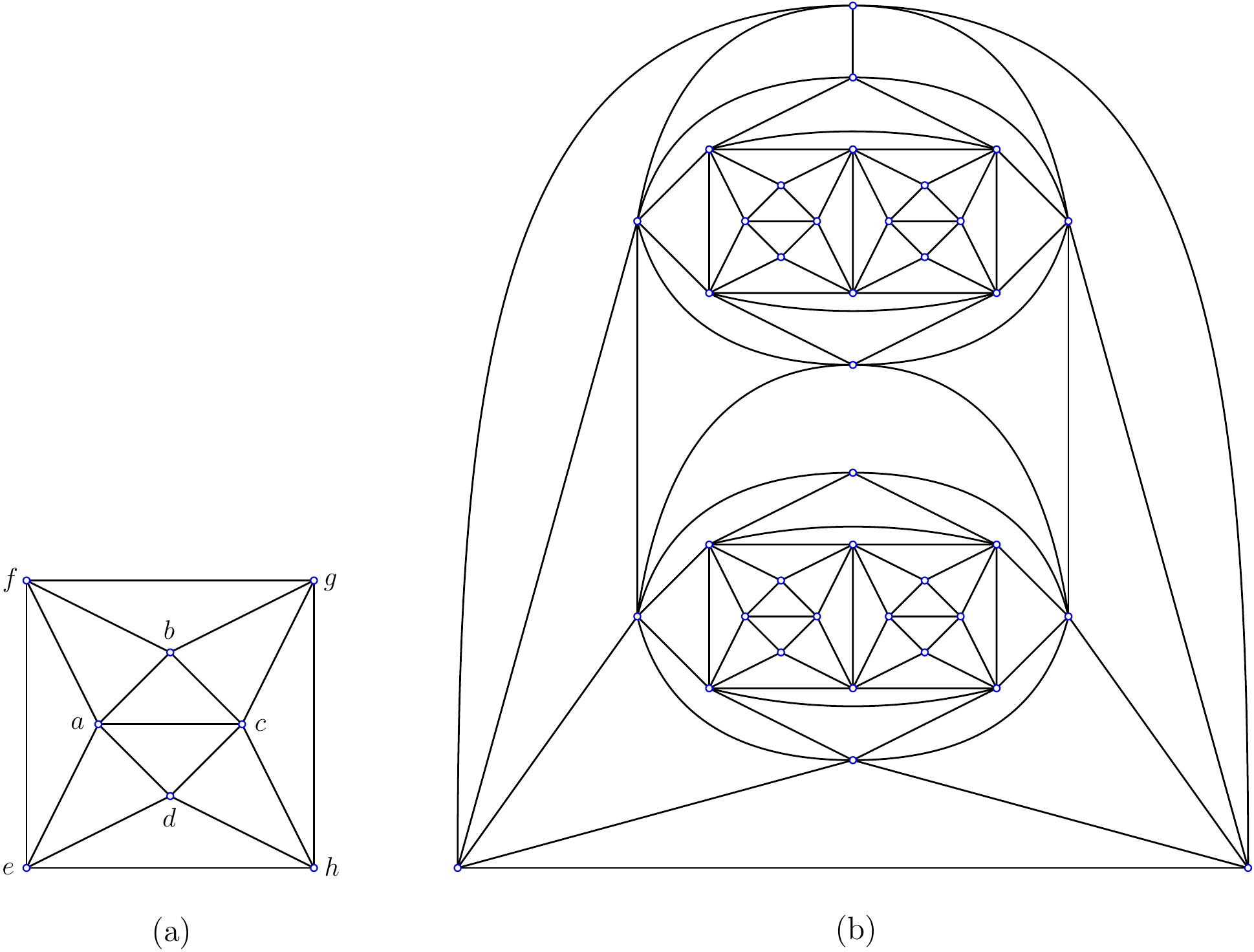}
\caption{(a) Graph $H$ in supporting the proof of Lemma~\ref{lem:orderingOfHVertices}. (b) Graph $G$ in supporting the proof of Theorem~\ref{thm:existsA4Connected}.}
\label{fig:4connectedExample}
\end{figure}

\subsection{Higher-$k$ $\delta_k$-Obstacle Representations}
\label{subsec:higherK}
In this section, we consider the non-crossing embeddings for $k>2$. We start by planar 3-trees.

\begin{theorem}
\label{thm:6grid3Trees}
Every planar 3-tree has a plane $\delta_3$-obstacle embedding.
\end{theorem}
\begin{proof}
The proof is by induction on $n=|V(G)|$. However, our inductive hypothesis is slightly stronger: every $n$ vertex planar 3-tree has a plane $\delta_3$-obstacle embedding in which the neighbours of each vertex $u$ occupy at least 3 of the sectors $Q^3_0(u),\dots,Q^3_5(u)$. As the base case, the smallest planar 3-tree is the clique $K_4$ on four vertices, for which the standard planar drawing of $K_4$ satisfies the requirement. So, assume that $n>4$ and every graph $G$ has a plane $\delta_3$-obstacle embedding that satisfies the stronger induction hypothesis for all $|V(G)|=k$ for all $k<n$.

When specialized to planar 3-trees, Lemma~\ref{lem:dujwood} says that every planar 3-tree is either $K_4$ or has a vertex $u$ and an independent set $S$ ($|S|\leq 3$) such that $G\setminus S$ is a 3-tree, $u$ has degree 3 in $G\setminus S$ with neighbours $x, y$ and $z$, and every vertex $r$ in $S$ forms a clique with exactly one of $uxy, uyz$ or $uzx$.

In the case $n>4$, we apply the previous result and recurse on $G\setminus S$. This give us a plane $\delta_3$-obstacle embedding of $G\setminus S$. By our (stronger) induction hypothesis, there are two cases depending on the locations of $x, y$ and $z$ with respect to $u$. In both cases, the elements of $S$ are placed close enough to $u$ that we do not create any new $\delta_3$-monotone paths involving vertices other than those in $\{u,x,y,z\}\cup S$. Furthermore, since $\{u,x,y,z\}$ form a clique, we only need to worry about (possibly) creating a new $\delta_3$-monotone path involving at least one vertex of $S$. We now consider two cases.
\begin{itemize}
\item No two neighbours of $u$ are in consecutive cones; e.g., $x\in Q^3_1(u), y\in Q^3_3(u)$ and $z\in Q^3_5(u)$. In this case, we add the elements of $S$ as shown in Figure~\ref{fig:6grid3Trees}(a).
\item Two neighbours of $u$ are in consecutive cones; e.g., $x\in Q^3_1(u), y\in Q^3_2(u)$ and $z\in Q^3_4(u)$. Then, we add the elements of $S$ as shown in Figure~\ref{fig:6grid3Trees}(b).
\end{itemize}
In both cases, we can verify that the (at most three) new neighbours of $u$ also satisfy the stronger inductive hypothesis.
\end{proof}

\begin{figure}[t]
\centering
\includegraphics[width=0.80\textwidth]{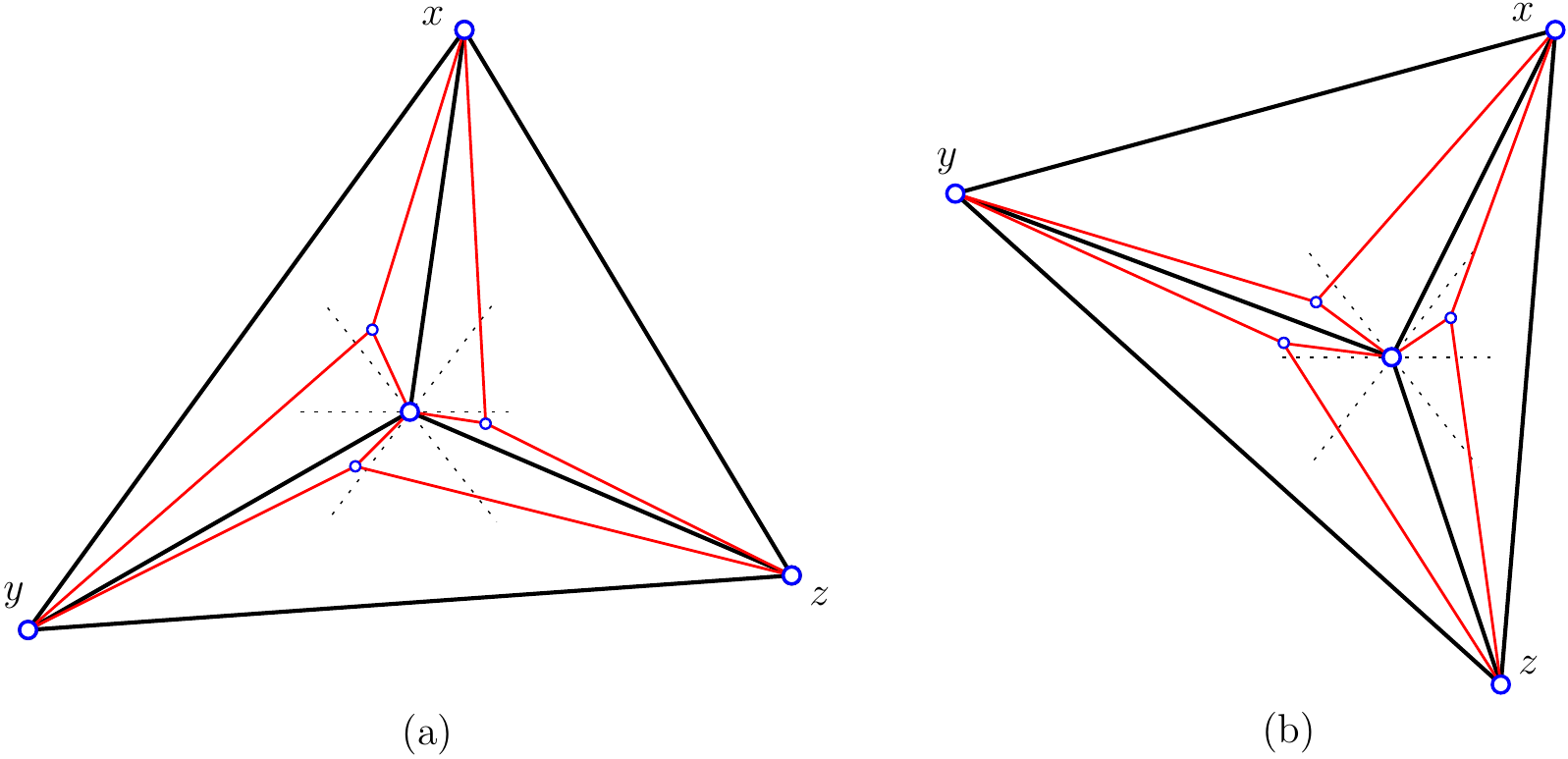}
\caption{An illustration in support of the proof of Theorem~\ref{thm:6grid3Trees}.}
\label{fig:6grid3Trees}
\end{figure}

\subparagraph{3-connected cubic graphs.} Here, we show that every $3$-connected cubic planar graph has a plane $\delta_7$-obstacle embedding. The algorithm contrustructs a $\delta_7$-obstacle embedding by adding one vertex per time according to a canonical ordering of the graph~\cite{Kant96}, and at each step it maintains a set of geometric invariants which guarantee its correctness. The key ingredients are the fact that each new vertex $v$ to be inserted has exactly two neighbors in the already constructed representation, together with the existence of a set of edges whose removal disconnects the representation in two parts, each containing one of the two neighbors of $v$. A sufficient stretching of these edges allows for a suitable placement for vertex $v$. We next give the details.

A graph is \emph{cubic} if all its vertices have degree three (i.e., it is $3$-regular). Let $G=(V,E)$ be a $3$-connected plane graph; i.e., a $3$-connected planar graph with a prescribed planar embedding. Let $\delta = \{\V_1,\dots,\V_K\}$ be an ordered partition of $V$, that is, $\V_1 \cup \dots \cup \V_K = V$ and $\V_i \cap \V_j = \emptyset$ for $i \neq j$. Let $G_i$ be the subgraph of $G$ induced by $\V_1 \cup \dots \cup \V_i$ and denote by $C_i$ the outerface of $G_i$. The partition $\delta$ is a \emph{canonical ordering} of $G$ if
\begin{itemize}
\item $\V_1=\{v_1,v_2\}$, where $v_1$ and $v_2$ lie on the outerface of $G$ and $(v_1,v_2) \in E$.
\item $\V_K = \{v_n\}$, where $v_n$ lies on the outerface of $G$, $(v_1,v_n) \in E$, and $v_n \neq v_2$.
\item Each $C_i$ ($i > 1$) is a cycle containing $(v_1,v_2)$.
\item Each $G_i$ is $2$-connected and internally $3$-connected.
\item For each $i = \{2, \dots, K-1\}$, one of the following conditions holds:
\begin{enumerate}
\item $\V_i$ is a \emph{singleton} $v_i$ which belongs to $C_i$ and has at least one neighbor in $G \setminus G_i$.
\item $\V_i$ is a \emph{chain} $\{v_i^1,\dots, v_i^l\}$,  both $v_i^1$ and $v_i^l$ have exactly one neighbor each in $C_{i-1}$, and $v_i^2, \ldots, v_i^{l-1}$ have no neighbor in $C_{i-1}$. Since $G$ is $3$-connected, this implies that each $v_i^j$ has at least one neighbor in $G \setminus G_i$.
\end{enumerate}
\end{itemize}

Kant~\cite{Kant96} proved that every $3$-connected plane graph has a canonical ordering that can be computed in linear time. Observe that, if the graph $G$ is  cubic and $\V_i$ is a singleton, then $v_i$ has exactly two neighbors in $C_i$ and therefore exactly one neighbor in $G \setminus G_i$. Similarly, if $\V_i$ is a chain, then all its vertices will have exactly one neighbor in $G \setminus G_i$, since they already have two neighbors in $G_i$. Therefore, for each $\V_i$, $i=2,\dots,K-1$, there are exactly two vertices in $G_{i-1}$ that are adjacent one to $v_i^1$ and one to $v_i^l$ if $\V_i$ is a chain, or both to $v_i$ if $\V_i$ is a singleton. We call them the \emph{leftmost} and the \emph{rightmost predecessor} of $\V_i$, respectively. 

Let $G$ be a $3$-connected cubic plane graph and let $\delta = \{\V_1,\dots,\V_K\}$ be a canonical ordering of $G$. For $i=2,\dots,k$, we call \emph{base edges} of $G_i$, the edges in the set $B_i$ inductively defined as follows: for $i=2$,  $B_2$ contains all edges of $G_2$; for $i>2$, if $\V_i$ is a singleton, then $B_i=B_{i-1}$, else $B_i=B_{i-1} \cup \{(v_i^1,v_i^2),\dots, (v_i^{l-1}v_i^l)\}$. Furthermore, every vertex of $G_i$ that has a neighbor in $G \setminus G_i$ is called an \emph{attaching vertex of $G_i$}. Note that an attaching vertex of $G_i$ belongs to $C_i$ and that $G_K$ has no attaching vertices. Two attaching vertices $u$ and $v$ of $G_i$ are \emph{consecutive} if there is no attaching vertex between them when walking from $u$ to $v$ along $C_i$ in the direction that does not pass through $v_1$ and $v_2$. We first need the following two results by Di Giacomo et al.~\cite{DiGiacomoLM18}.
\begin{lemma}[Di Giacomo et al.~\cite{DiGiacomoLM18}]
\label{lem:baseEdges}
For every pair of consecutive attaching vertices $u$ and $v$ of $G_i$, there exists a set of base edges $B_i(u,v)$ whose removal disconnects $G_i$ into two subgraphs, one containing $u$ and the other one containing $v$.
\end{lemma}

All other edges of $G_i$ that are not base edges are called \emph{attaching edges}.

\begin{lemma}[Di Giacomo et al.~\cite{DiGiacomoLM18}]
\label{lem:exactlyOneBaseEdge}
For every pair of consecutive attaching vertices $u$ and $v$ of $G_i$, there exists exactly one base edge in the path from $u$ to $v$ along $C_i$, and thus all other edges in this path (if any) are attaching edges.
\end{lemma}

We are now ready to prove the following.

\begin{theorem}
\label{thm:3connectedCubic}
Every $3$-connected cubic plane graph has a plane $\delta_7$-obstacle embedding.
\end{theorem}
\begin{proof}
Let $\delta = \{\V_1,\dots,\V_K\}$ be a canonical ordering of a $3$-connected cubic plane graph $G$. The algorithm inductively constructs a drawing of $G$ by adding a set $\V_i$ per time. The base case is a drawing of $G_2$. We denote by $\Gamma_{i}$ the drawing after the addition of $\V_i$ ($i=2,3,\dots,K$); i.e., the drawing of $G_i$. We prove that each drawing $\Gamma_i$ ($i=2,\dots,K-1$) satisfies the following invariants:
\begin{itemize}
\item[\textbf{I1.}] Every base edge $(u,v)$, assuming $u$ is below $v$, is such that $v$ is inside $Q^7_0(u)$ (resp., $Q^7_6(u)$) if $v$ is to the right (resp., left) of $u$.
\item[\textbf{I2.}] Every other edge $(u,v)$, assuming $u$ is below $v$, is such that $v$ is inside $Q^7_1(u)$ or $Q^7_2(u)$ (resp., $Q^7_4(u)$ or $Q^7_5(u)$) if $v$ is to the right (resp., left) of $u$.
\item[\textbf{I3.}] For every two consecutive attaching vertices $u$ and $v$, assuming $u$ is to the left of $v$, the path from $u$ to $v$ along $C_i$ is drawn $x$-monotone and it first contains a (possibly empty) set of downward attaching edges, then the unique base edge (in Lemma~\ref{lem:exactlyOneBaseEdge}), and then a (possibly empty) set of upward attaching edges.
\end{itemize}

\begin{figure}
\centering 
    \begin{minipage}[b]{.36\textwidth}
        \centering
        \includegraphics[width=\textwidth,page=1]{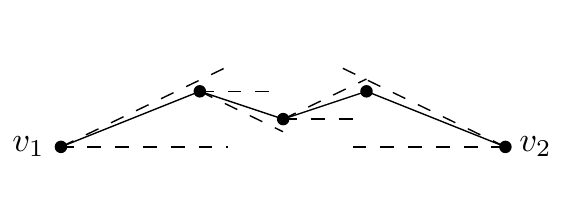}
        \subcaption{~}\label{fig:cubic-g2-odd}{} 
    \end{minipage} 
	\hfil
	\begin{minipage}[b]{.36\textwidth}
		\centering 
		\includegraphics[width=\textwidth,page=2]{cubic}
		\subcaption{~}\label{fig:cubic-g2-even}{}
	\end{minipage}
	\begin{minipage}[b]{.36\textwidth}
        \centering
        \includegraphics[width=\textwidth,page=3]{cubic}
        \subcaption{~}\label{fig:cubic-vi-before}{} 
    \end{minipage} 
	\hfil
	\begin{minipage}[b]{.36\textwidth}
		\centering 
		\includegraphics[width=\textwidth,page=4]{cubic}
		\subcaption{~}\label{fig:cubic-vi-after}{}
	\end{minipage}
	\caption{(a-b) Illustrations for the drawing of $G_2$ when $|\V_2|$ is (a) odd and (b) even. (c-d) Illustrations for the addition of $\V_i$ (c) before the stretching operation and (d) after the stretching operation.}
\end{figure}

We first draw $G_2$ as follows. We distinguish between two cases based on whether $\V_2$ contains an odd or even number of vertices. This distinction is needed to avoid monotone paths and to ensure visibility between $v_1$ and $v_2$, refer to Figures~\ref{fig:cubic-g2-odd} and~\ref{fig:cubic-g2-even} for illustrations. If $|\V_2|$ is odd, we set $v_1$ at point $(0,0)$ and $v_2$ at point $(1,0)$. We then draw each vertex $v_1^j$ ($j=1,\dots,l$) in the intersection of $Q^7_0(v_1)$ and $Q^7_7(v_2)$. Each vertex $v_1^j$ is placed inside the sector $Q^7_0(v_1^{j-1})$ if $j$ is odd, or inside the sector $Q^7_{13}(v_1^{j-1})$ if $j$ is even, where $v_1^{0}=v_1$. Also, we may assume that all vertices with even (odd) $j$ have the same $y$-coordinate. If $|\V_2|$ is even, we set $v_1$ at point $(0,0)$ and $v_2$ at point $(1,-\epsilon)$, for a value of $\epsilon$ that guarantees $v_2$ be inside $Q^7_{13}(v_1)$. We then draw each vertex $v_1^j$ ($j=1,\dots,l$) below $v_1$, above the line connecting $v_1$ and $v_2$, and between $v_1$ and $v_2$. Each vertex $v_1^j$ is placed inside the sector $Q^7_{13}(v_1^{j-1})$ if $j$ is odd, or inside the sector $Q^7_0(v_1^{j-1})$ if $j$ is even, where $v_1^{0}=v_1$. Also, we may assume that all vertices with even (odd) $j$ have the same $y$-coordinate. By definition all edges of $G_2$ are base edges (thus I2. trivially follows), and this construction guarantees I1. and I3.

Next, we draw each $\V_i$ as follows, refer to Figures~\ref{fig:cubic-vi-before} and~\ref{fig:cubic-vi-after} for illustrations. Let $u$ and $v$ be the leftmost and rightmost predecessors of $\V_i$, respectively. One can see that, by the properties of the canonical ordering together with I2, if $u$ is adjacent to a vertex $w$ in $G_{i-1}$ and $u$ is in $Q^7_1(w)$ (resp., $Q^7_2(w)$), then there is no vertex $z$ in $G_{i-1}$ that is adjacent to $u$ and such that $u$ is in $Q^7_2(w)$ (resp., $Q^7_1(w)$). Hence, we can place the vertex of $\V_i$ adjacent to $u$ in either $Q^7_1(u)$ or $Q^7_2(u)$, say $Q^7_1(u)$, without creating any monotone path. Similarly, we can place the vertex of $\V_i$ adjacent to $v$ inside either $Q^7_4(v)$ or $Q^7_5(v)$, say $Q^7_4(v)$. Indeed, we can place all vertices of $\V_i$ in $Q^7_1(u)\cap Q^7_4(v)$. To this aim, we need to ensure that  $Q^7_1(u)\cap Q^7_4(v)\neq \emptyset$ and that $Q^7_1(u)\cup Q^7_4(v)$ does not intersect $\Gamma_{i-1}$. By I3., the path between $u$ and $v$ contains first a set of downward edges, followed by one base edge, and finally a set of upward edges.  Thus, by sufficiently stretching the base edges in $B_{i-1}(u,v)$  (Lemma~\ref{lem:baseEdges}) both conditions can be satisfied. I2. is trivially preserved by the stretching operation. Concerning I1 and I3, note that stretching a base edge makes it nearly horizontal but does not change the sector in which the edge lies. After the stretching operation, we can safely draw $\V_i$ if it is a singleton in $Q^7_1(u)\cap Q^7_4(v)$, or by using a technique similar as for $G_2$ if it is a chain.

Invariants I1. and I2. imply that $G_{K-1}$ admits a plane $\delta_7$-obstacle embedding. The addition of $\V_K$ requires some special care since $v_K$ has three predecessors. One of these predecessors is $v_1$ by definition of canonical ordering. Let $u$ and $v$ be the other two predecessors in the order they appear  when walking clockwise along $C_{K-1}$ from $v_1$ to $v_2$. We place $v_K$ in $Q^7_3(u)\cap Q^7_4(v)$ similarly as we did for the previous sets (this may require stretching the base edges in $B_{K-1}(u,v)$). We now aim at realizing the edge $(v_1,v_K)$ such that $v_K$ is in $Q^7_2(u)$ (or equivalently in $Q^7_3(u)$). If $v_K$ already belongs to $Q^7_2(u)$, then we are done. Else, we can assume that the $y$-coordinate of $v_K$ is sufficiently large to guarantee that $v_K$ lies above $Q^7_2(u)$. Note that both the other two edges incident to $v_1$ are base edges that belong to $G_2$, thus we can move $v_1$ horizontally to the left until $v_K$ lies in $Q^7_2(u)$. This concludes the proof.
\end{proof}

\subparagraph{Bounded-degree planar graphs.} Here, we show that every planar graph with maximum degree $\Delta$ has a plane $\delta_{O(\Delta)}$-obstacle embedding. Informally speaking, the idea is to apply the algorithm of Keszegh et al.~\cite{KeszeghPP13} for drawing bounded-degree planar graphs with a few slopes and then taking theses (bounded) number of slopes to define a polyhedral distance function with $k=O(\Delta)$. However, there might exist edges in this drawing with collinear endpoints, which are not allowed for our purpose. We resolve this by re-scaling the drawing and a 4-colouring of the graph. We next give the details.

Let $G$ be a planar graph with maximum degree $\Delta$ and let $c=c(\Delta)$ be an integer that depends only on $\Delta$. Keszegh et al.~\cite{KeszeghPP13} proved that we can assign to each vertex $u\in V(G)$ is assigned a non-negative integer $\ell(u)$ such that (i) for any edge $uw\in E(G)$, $|\ell(u)-\ell(w)|\leq c$, and (ii) there exists a planar straight-line drawing of $G$ in which the $x$- and $y$-coordinates of every vertex $u$ are both divisible by $2^{\ell(u)-c}$, and each edge incident to a vertex $u$ has length at most $2^{\ell(u)+c}$. Notice that each edge in this drawing has a slope that is the same as that of some line segment whose endpoints are on a $2^{2c}\times 2^{2c}$ grid; this drawing thus uses a fixed number of slopes, for a fixed $\Delta$.

It remains to show how the collinear edges are dealt with. First, we compute a 4-colouring $c: V(G)\to \{0,1,2^c+2,2^{2c}+2^c+4\}$ of the vertices of $G$. Next, we re-scale the grid to a $2^{6c}\times 2^{6c}$ grid. Then, for every vertex $u=(x_u,y_u)\in V(G)$, we move $u$ to the grid point with coordinates $(x_u,c(u))$; i.e., the $y$-coordinate of every vertex is increased by the value of its colour. In the following, we show that this transformation will result in a drawing of $G$ with no collinear edges. For a vertex $u\in V(G)$, let $u'$ denote its new position.

Consider two edges $uv$ and $vw$ that were collinear before this transformation and assume w.l.o.g. that $x_u<x_v<x_w$. Therefore, $(x_v-x_u)/(y_v-y_u)=(x_w-x_v)/(y_w-y_v)$; assume w.l.o.g. that $x_w-x_v=a(x_v-x_u)$ and $y_w-y_v=a(y_v-y_u)$, where $a$ is a rational with both enumerator and denominator between 1 and $2^c$. If these two edges are still collinear after the transformation, then we must have
\[
\frac{x_v-x_u}{(y_v-y_u)+(c(v)-c(u))}=\frac{x_w-x_v}{(y_w-y_v)+(c(w)-c(v))},
\]
which is simplified to $c(w)-c(v)=a(c(v)-c(u))$. But, it is easy to verify by the bound on $a$ and the values given to the colours that this equality is not possible.

Now, consider two edges $uv$ and $vw$ that were non-collinear before this transformation, and assume w.l.o.g. that $x_u<x_v<x_w$. We can argue that they stay non-collinear after the transformation by noting a lower bound on the difference between their slopes. That is, we know that
\[
|\frac{x_v-x_u}{y_v-y_u}-\frac{x_w-x_v}{y_w-y_v}|>\frac{1}{2^{2c}}.
\]
Since the grid is re-scaled by $2^{6c}$ and the largest value of a colour assigned to a vertex is $2^{2c}+2^c+4$, we can ensure that these two edges cannot become collinear after the transformation. Hence, we have the following theorem.
\begin{theorem}
Every planar graph with maximum degree $\Delta$ has a plane $\delta_{O(\Delta)}$-obstacle embedding.
\end{theorem}

\section{Graph Metrics}
\label{sec:graphMetric}
In this section, we consider the problem under graph distances. Recall the graph $D$-cube, $Q_D$ whose vertex set is $V(Q_D)=\{0,1\}^D$ and that contains the edge $uw$ if and only $u$ and $w$ differ in exactly one coordinate. It is not hard to see that every $n$ vertex graph has a $Q_n$-obstacle representation: each vertex of $G$ is assigned a coordinate with a single 1 bit. Then, for any two vertices $u$ and $w$ there are exactly two shortest paths in $Q_n$ joining them and they each have length 2. One path goes through the intermediate vertex $\mathbf{0}=(0,\ldots,0)$ and the other goes through $u+w$. Therefore, by placing an obstacle at $\mathbf{0}$ and at each $u+w$ for which $uw\not\in E(G)$, we obtain a $Q_D$-obstacle representation of $G$. The following theorem shows we can do this with much fewer coordinates.
\begin{theorem}
\label{thm:hypercube}
There exists a constant $C>0$ such that, for $D=C\log n$, every $n$-vertex graph has a non-crossing $Q_D$-obstacle representation.
\end{theorem}
\begin{proof}
Consider the following embedding $(\varphi,c)$ of $G$ into $Q_D$:  For
each $u\in V(G)$, $\varphi(u)$ is a random element of $\{0,1\}^D$.  We use
the notation $u_i$ to denote the $i$th coordinate of $u$.  Let $\prec$
denote lexicographic order on $D$-tuples.  For each edge $uw\in E(G)$
with $u\prec w$, we take $c(uw)$ to be the
   \emph{greedy} path that visits, for $i=0,\ldots,D$, the vertex
   $uw_i=(w_1,\ldots,w_i,u_{i+1},\ldots,u_D)$.  Thus $uw_0,\ldots,uw_D$
   is a sequence of vertices that---after removing duplicates---is a
   shortest path, in $Q_D$, from $u$ to $w$.  Note that there is an
   asymmetry here that we should be careful of, so for $u\prec w$,
   we define $wu_i=w_1,\ldots,w_{D-i},u_{D-i+1},\ldots,u_{D}=uw_{D-i}$.
   Here are some observations about the embedding $(\varphi,c)$:
   \begin{enumerate}
      \item All vertex distances are close to $D/2$:
       The distance between any two vertices is a
       binomial$(D,1/2)$ random variable.  Therefore, by Chernoff's
       bounds, for any constant $\eps>0$ and for any vertex pair $u\neq
       w$, $\Pr\{|\delta_{Q_D}(u,w)-D/2| > \eps(D/2)\} \le n^{-\Omega(C)}$.  By the
       union bound, the probability that there exists any pair of vertices
       $u\neq w$ with $|\delta_{Q_D}(u,w)-D/2| > \eps(D/2)$ is also $n^{-\Omega(C)}$.

      \item The embedding is non-crossing: For any four distinct vertices
      $u\prec w$ and $x\prec y$, and any $i,j\in\{0,\ldots,D\}$, the
      vertices $uw_i$ and $xy_j$ are independent random $D$-bit strings.
      Therefore, $\Pr\{\delta_{Q_D}(uw_i,xy_j)\le 1\} = (D+1)/2^{D}$.
      By the union bound, the probability that there exists any four
      vertices $u,w,x,y$ and any pair of indices $i,j$ for which 
      $\delta_{Q_D}(uw_i,xy_j)\le 1$ is at most 
      $n^4(D+1)^3/2^{D}=n^{-\Omega(C)}$.

     \item No geodesic passes close to a vertex except its endpoints:
      Let $u$, $w$, and $x$ be distinct vertices and $r\in\{0,\ldots,D\}$
      be an integer.  Then, the probability that there exists any geodesic
      with endpoints $u$ and $w$ that contains a vertex $z$ with
      $\delta_{Q_D}(z,x)\le r$ is at most $n^{-\Omega(c)}$.  To see why this
      is so, suppose that such a geodesic, $C$, contains a vertex $z$
      such that $\delta_{Q_D}(z,x) \le r$. Then, at least one of the following
      events occurs: \begin{inparaenum}[(a)] \item $\delta_{Q_D}(u,w) \ge (1+\eps)D/2$; \item $\delta_{Q_d}(u,x) \le (1+\eps)D/4+r$; or \item $\delta_{Q_d}(w,x) \le (1+\eps)D/4+r$. \end{inparaenum}
      
     Point 1, above establishes that the probability of the first event is
     $n^{-\Omega(c)}$ and that, for $r\le (1-3\eps)D/4$, the probability
     of each of the other two events is $n^{-\Omega(c)}$.  Applying the
     union bound over all 3 events, and over all $\binom{n}{3}$ choices
     of $u$, $w$, and $x$ then shows that the probability that there
     is any triple $u$, $w$, $x$ such that any geodesic from $u$ to $w$
     passes within distance $(1-3\eps)D/4$ of $x$ is $n^{-\Omega(c)}$.

     \item Paths diverge quickly:
      Let $xu,xw\in E(G)$, be two edges of $G$ with the common endpoint
      $x$ and let $r\in\{0,\ldots,D\}$.  We want to show that the
      directed paths $xu$ and $xw$ diverge quickly.  There are three
      cases to consider:
     \begin{enumerate}
        \item $x\prec u$ and $x\prec w$.  In this case $xu_r=xw_r$ if and
        only if $u_1,\ldots,u_r=w_1,\ldots,w_r$, so $\Pr\{xu_r=xw_r\}=
        2^{-r}$.

        \item $x\prec u$ and $w\prec x$.  In this case, we
        consider $xu_r=u_1,\ldots,u_r,x_{r+1},\ldots,x_D$ and
        $xw_r=wx_{D-r}=x_1,\ldots,x_{D-r},w_{D-r+1},\ldots,w_D$. For
        any choice of $i$, these two strings have independent bits in
        at least $r$ locations, so $\Pr\{xu_r=xw_r\}\le 2^{-r}$.

        \item $u\prec x$ and $w\prec x$. In this case
        $xu_r=ux_{D-r}=x_1,\ldots,x_{D-r},u_{D-r+1},\ldots,u_D$ and
        $xw_r=wx_{D-r}=x_1,\ldots,x_{D-r},w_{D-i+1},\ldots,w_D$. So
        $\Pr\{xu_r=xw_r\}=2^{-r}$.
\end{enumerate}
     If we choose $r=\alpha\log n$, then this probability is at most
     $n^{-\Omega(\alpha)}$. Again, the union bound shows that the probability
     that there is any $u$, $w$, or $x$ such that $xu_r=xw_r$ is at
     most $n^{-\Omega(\alpha)}$.
    \end{enumerate}
    In the following, we choose $C$ sufficiently large and $\alpha <
    (1/4-\epsilon)C$ also sufficiently large so that with probability
    greater than 0, we obtain an embedding for which all four of preceding
    properties hold.  Therefore, there exists some embedding $(\varphi,
    c)$ such that 1. for all $u,w\in V(G)$, $|\delta_{Q_D}(u,w)-D/2|\le \epsilon D/2$; 2. for all $uw,xy\in E(G)$ with $\{u,w\}\cap\{x,y\}=\emptyset$, $\delta_{Q_D}(c(uw),c(xy)) > 1$; 3. for all $uw\in E(G)$ and $x\in V(G)\setminus\{u,w\}$, $\delta_{Q_D}(c(uw),x) \ge (1-\epsilon)D/4$; and 4. for all $xu,xw\in E(G)$ and all $r\ge\alpha\log n$, $xu_r\neq xw_r$.

    To obtain a $Q_D$-obstacle representation $(\varphi,S)$ we take $S$
    to contain all the vertices not used in any path of the embedding
    $(\varphi, c)$. To verify that this is indeed a non-crossing
    $Q_D$-obstacle representation, we need only verify that, for any
    $u,w\in V(G)$ with $uw\not\in E(G)$, 
    $\delta_{Q_D\setminus S}(u,w)> \delta_{Q_D}(u,w)$.
    This is implied by the following inequality, which relates distances
    in $G$ to those in $Q_D\setminus S$:
    \begin{equation}
       \delta_{Q_D\setminus S}(u,w) \ge \delta_G(u,w)(1-\epsilon)D/2
        -(\delta_G(u,w)-1)2\alpha\log n \enspace. \label{eq:stretcher}
    \end{equation}
    Note \eqref{eq:stretcher} is sufficient since, if $uw\not\in E(G)$, then $\delta_G(u,w)\ge 2$ and \eqref{eq:stretcher} implies $\delta_{Q_D\setminus S}(u,w) \ge (1-\epsilon)D-2\alpha\log n = ((1-\epsilon)C - 2\alpha)\log n > (1+\epsilon)D/2$,
    which contradicts Property~1.   Thus, all that remains is to establish
    \eqref{eq:stretcher}.
    To do this, consider any path $P$ from $u$ to $w$ in $Q_D\setminus S$.
    Since the only vertices in $Q_D\setminus S$ are those that are used
    by some embedded edge of $G$, the path $P$ consists of a sequence
    of subpaths $P_0,\ldots,P_k$ where each $P_i$ is a subpath of
    $c(x_iy_{i})$ for some edge $x_{i}y_{i}\in E(G)$.  Note that
    Property~3 implies that $x_0=u$ and that $x_k=w$. Furthermore,
    Properties~2 and 3 imply that $x_{i}=y_{i-1}$ for each
    $i\in\{1,\ldots,k\}$.  Therefore, $x_0,\ldots,x_k$ is a path in $G$
    from $u$ to $w$, so $k\ge \delta_G(u,w)$.  Finally, Property~4 implies
    that, for each $i\in\{1,\ldots,k-1\}$, the portion of $c(x_i,x_{i+1})$
    not used by $P_i$ has length at most $2\alpha\log n$. Thus, the
    length of $P$ is at least $k(1-\epsilon)D/2 - 2(k-1)\alpha\log n$,
    as required.
\end{proof}

It is worth noting that Theorem~\ref{thm:hypercube} is closely related to Theorem~\ref{thm:3d-universal}.  Indeed, before perturbing it, the point set $X$ used in the proof of Theorem~\ref{thm:3d-universal} is a projection of the vertices of $Q_{D}$ with $D=\lceil\log_2 n\rceil$ onto $\R^3$. In Theorem~\ref{thm:3d-universal} we then perturb $X$ to obtain a non-crossing embedding. In the proof of Theorem~\ref{thm:hypercube} we have to be more careful to avoid crossings.

\section{Conclusion}
\label{sec:conclusion}
In this paper, we introduced the geodesic obstacle representation of graphs, providing a unified generalization of obstacle representations and grid obstacle representations. Our work leaves several problems open. As perhaps the main question, does every planar graph admit a non-crossing $\delta_k$-obstacle representation for some constant $k$? It would be also interesting to extend the classes of graphs for which non-crossing $\delta_k$-obstacle representations exist for small values of $k$. For graph metrics, given two graphs $G$ and $H$, is it \textsc{NP}-hard to decide if $G$ has an $H$-obstacle representation?

\subparagraph{Acknowledgement.} We thank Mark Keil for useful discussions on this problem.

\bibliography{ref}

\begin{thebibliography}{10}

\bibitem{AlonKS03}
Noga Alon, Michael Krivelevich, and Benny Sudakov.
\newblock Tur\'{a}n numbers of bipartite graphs and related {R}amsey-type
  questions.
\newblock {\em Combinatorics, Probability {\&} Computing}, 12(5-6):477--494,
  2003.

\bibitem{AKL10}
Hannah Alpert, Christina Koch, and Joshua~D. Laison.
\newblock Obstacle numbers of graphs.
\newblock {\em Discrete {\&} Computational Geometry}, 44(1):223--244, 2010.

\bibitem{BalkoCV18}
Martin Balko, Josef Cibulka, and Pavel Valtr.
\newblock Drawing graphs using a small number of obstacles.
\newblock {\em Discrete {\&} Computational Geometry}, 59(1):143--164, 2018.

\bibitem{BermanCFGHW17}
Leah~Wrenn Berman, Glenn~G. Chappell, Jill~R. Faudree, John Gimbel, Chris
  Hartman, and Gordon~I. Williams.
\newblock Graphs with obstacle number greater than one.
\newblock {\em J. Graph Algorithms Appl.}, 21(6):1107--1119, 2017.

\bibitem{BiedlM-GD17}
Therese~C. Biedl and Saeed Mehrabi.
\newblock Grid-obstacle representations with connections to staircase guarding.
\newblock In {\em proceedings of the 25th International Symposium on Graph
  Drawing and Network Visualization (GD 2017), Boston, MA, USA}, 2017.
\newblock arXiv version available at: arxiv.org/abs/1708.01903.

\bibitem{BGMMP}
Arijit Bishnu, Arijit Ghosh, Rogers Mathew, Gopinath Mishra, and Subhabrata
  Paul.
\newblock Grid obstacle representations of graphs, 2017.
\newblock arXiv version available at: arxiv.org/abs/1708.01765.

\bibitem{DujmovicM15}
Vida Dujmovic and Pat Morin.
\newblock On obstacle numbers.
\newblock {\em Electr. J. Comb.}, 22(3):P3.1, 2015.

\bibitem{DujmovicW07}
Vida Dujmovic and David~R. Wood.
\newblock Graph treewidth and geometric thickness parameters.
\newblock {\em Discrete {\&} Computational Geometry}, 37(4):641--670, 2007.

\bibitem{ErdosR59}
Paul Erd\H{o}s and Alfred Renyi.
\newblock On random graphs.
\newblock {\em Publicationes Mathematicae}, 6:290--297, 1959.

\bibitem{DiGiacomoLM18}
Emilio~Di Giacomo, Giuseppe Liotta, and Fabrizio Montecchiani.
\newblock Drawing subcubic planar graphs with four slopes and optimal angular
  resolution.
\newblock {\em Theoretical Computer Science, to appear}, 714:51--73, 2018.

\bibitem{GimbelMV17}
John Gimbel, Patrice~Ossona de~Mendez, and Pavel Valtr.
\newblock Obstacle numbers of planar graphs.
\newblock In {\em Graph Drawing and Network Visualization - 25th International
  Symposium, {GD} 2017, Boston, MA, USA, September 25-27, 2017, Revised
  Selected Papers}, pages 67--80, 2017.

\bibitem{JohnsonS14}
Matthew~P. Johnson and Deniz Sari{\"{o}}z.
\newblock Representing a planar straight-line graph using few obstacles.
\newblock In {\em proceedings of the 26th Canadian Conference on Computational
  Geometry, (CCCG 2014), Halifax, NS, Canada}, 2014.

\bibitem{Kant96}
Goos Kant.
\newblock Drawing planar graphs using the canonical ordering.
\newblock {\em Algorithmica}, 16(1):4--32, 1996.

\bibitem{KeszeghPP13}
Bal{\'{a}}zs Keszegh, J{\'{a}}nos Pach, and D{\"{o}}m{\"{o}}t{\"{o}}r
  P{\'{a}}lv{\"{o}}lgyi.
\newblock Drawing planar graphs of bounded degree with few slopes.
\newblock {\em {SIAM} J. Discrete Math.}, 27(2):1171--1183, 2013.

\bibitem{MukkamalaPP12}
Padmini Mukkamala, J{\'{a}}nos Pach, and D{\"{o}}m{\"{o}}t{\"{o}}r
  P{\'{a}}lv{\"{o}}lgyi.
\newblock Lower bounds on the obstacle number of graphs.
\newblock {\em Electr. J. Comb.}, 19(2):P32, 2012.

\bibitem{PachS11}
J{\'{a}}nos Pach and Deniz Sari{\"{o}}z.
\newblock On the structure of graphs with low obstacle number.
\newblock {\em Graphs and Combinatorics}, 27(3):465--473, 2011.

\bibitem{Valtr97}
Pavel Valtr.
\newblock Graph drawings with no k pairwise crossing edges.
\newblock In {\em Graph Drawing, 5th International Symposium, {GD} '97, Rome,
  Italy, September 18-20, 1997, Proceedings}, pages 205--218, 1997.

\end{thebibliography}

\end{document}